\documentclass[a4paper, 12pt, reqno]{amsart}

\numberwithin{equation}{section}

\usepackage{amssymb}
\usepackage{mathrsfs}
\usepackage{dsfont}
\usepackage{stmaryrd, MnSymbol}
\usepackage{enumerate, xspace}
\usepackage{changebar}
\usepackage{hyperref}
\usepackage{pdfsync}

\newtheorem{theorem}{Theorem}
\newtheorem{lemma}{Lemma}[section]
\newtheorem{corollary}{Corollary}[section]
\newtheorem{proposition}{Proposition}[section]

\newtheorem{remark}{Remark}[section]

\numberwithin{equation}{section}

\newcommand{\de}{\delta}

\newcommand{\Ga}{\Gamma}
\newcommand{\La}{\Lambda}

\def\R{{\mathbb{R}}}
\def\N{{\mathbb{N}}}
\def\Z{{\mathbb{Z}}}
\def\T{{\mathbb{T}}}

\def\Cy{{\mathcal{C}}}
\def\D{{\mathcal{D}}}

\def\E{{\mathcal{E}}}
\def\G{{\mathcal{G}}}
\def\F{{\mathcal{F}}}
\allowdisplaybreaks

\newcommand{\triple}{|\!|\!|}

\begin{document}
\title[Energy Diffusion]
{Macroscopic energy diffusion for a chain of anharmonic
oscillators}
\author{Stefano Olla}
\address{Stefano Olla\\
CEREMADE, UMR CNRS 7534\\
Universit\'e Paris-Dauphine\\
75775 Paris-Cedex 16, France \\
}
\email{{\tt olla@ceremade.dauphine.fr}}
\author{Makiko Sasada}
\address{Makiko Sasada\\
Department of Mathematics, Keio University\\
3-14-1, Hiyoshi, Kohoku-ku, Yokohama-shi, Kanagawa, 223-8522, Japan}
\email{{\tt sasada@math.keio.ac.jp}}
\date{\today}
\begin{abstract}
We study the energy diffusion in a chain of anharmonic
oscillators where the Hamiltonian dynamics is perturbed by a local energy
conserving noise. We prove that under diffusive rescaling of
space-time, energy fluctuations diffuse and evolve following an
infinite dimensional linear stochastic differential equation driven by
the linearized heat equation. We also give variational expressions for the
thermal diffusivity and some upper and lower bounds. 
\end{abstract}
\thanks{ This paper has been partially supported by the
  European Advanced Grant {\em Macroscopic Laws and Dynamical Systems}
  (MALADY) (ERC AdG 246953), by Agence Nationale de la Recherche, 
  under grant ANR-2010-BLAN-0108 (SHEPI).\\
  We thank Professor T. Funaki for insightful discussions and his interest in this work.} 
\subjclass[2000]{}

\maketitle

\section{Introduction}
\noindent
The deduction of the heat equation or the Fourier law for the macroscopic
evolution of the energy through a diffusive space-time scaling limit
from a microscopic dynamics given by Hamilton or Schr\"odinger
equations, is one of the most important problem in non-equilibrium
statistical mechanics (\cite{BRL}). One dimensional chains of
oscillators have been used as simple models for this study. In the context of the
classical (Hamiltonian) dynamics, it is clear that non-linear
interactions are crucial for the diffusive behavior of the energy. In
fact, in a chain of harmonic oscillators the energy evolution is
ballistic (\cite{RLL}). In this linear system, the energy of each mode
of vibration is conserved. Non-linearities introduce interactions
between different modes and destroy these conservation laws and give a certain
ergodicity to the microscopic dynamics.  

In order to describe the mathematical problem, let us introduce some
notation we will use in the rest of the paper.
We study a system of anharmonic oscillators, each is denoted by an
integer $i$. We denote by $(q_i, p_i)$ the corresponding position and
momentum (we set the mass equal to 1). Each pair of consecutive 
particles $(i, i + 1)$ are connected by a spring which can be anharmonic. The
interaction is described by a potential energy  
$\bar{V}(q_{i+1}- q_{i})$. 
We assume that $\bar{V}$ is a nonnegative 
smooth function satisfying 
\begin{displaymath}
Z_{\beta}:=\int_{\R}e^{-\beta \bar{V}(r)}dr < \infty
\end{displaymath}
for all $\beta > 0$. 
Let $a$ be the equilibrium inter-particle distance, where $\bar{V}$
attains its minimum that we assume to be $0: \bar{V} (a) = 0$. 
It is convenient to work with inter-particle distances as coordinates,
rather than absolute particle positions, so we 
define $r_j = q_j - q_{j-1} - a$. 
We denote the translated function $\bar{V}(\cdot + a )$ by $V(\cdot)$ hereafter. Namely, we
assume $V(0)=0$. 
The configuration of the system is given by $\{p_j , r_j\}_j$, and  
energy function (Hamiltonian) defined for each configuration is
formally given by
\begin{equation*}
H =\sum_{j} \E_j, \qquad \E_j = \frac{1}{2} p_j^2 + V (r_j). 
\label{eq:16}
\end{equation*}
The choice of $\E_j$ as the
energy of each oscillator is a bit arbitrary, 
because we associate the potential energy of the bond $V (r_j)$ to the
particle $j$. Different choices can be 
made, but this one is notationally convenient.

The corresponding Hamiltonian dynamics is given by the equations of
motion: 
\begin{equation}
\label{hequ}
\begin{split}
r'_j(t) &= p_j(t)- p_{j-1}(t),   \\
p'_j(t) &= V'(r_{j+1}(t))- V'(r_j(t)).
\end{split}
\end{equation}
We are interested in the macroscopic evolution of the empirical energy
profile under a diffusive macroscopic space-time scaling. More
precisely the limit, as $N\to\infty$, of the energy distribution on
$\mathbb R$ defined by 
\begin{equation}
  \label{eq:15}
  \frac 1N \sum_{i} \E_i(N^2 t) \delta_{i/N}(dy).
\end{equation}
Energy is not the only conserved quantity under the dynamics
\eqref{hequ}. Formally length and 
momentum are also integral of the motion.  
In one dimensional system, even for anharmonic interaction, generically we expect a
superdiffusion of the energy, essentially because of the momentum
conservation (\cite{LLP,BBO1}). Adding a pinning potential
$U(q_i)$ on each particle, it will break the translation
invariance of the system and the momentum conservation, and we expect
a diffusive behavior for the energy, i.e. the energy profile defined
by \eqref{eq:15} would converge to the solution $\E(t,y)$ of a heat equation
\begin{equation}
  \label{eq:nonl}
  \partial_t \E = \partial_y \left(D(\E) \partial_y \E\right)
\end{equation}
under specific conditions on the initial configuration. 
The diffusivity $D= D(\E)$ is defined by the Green-Kubo formula
associated to the corresponding infinite dynamics in equilibrium at
average energy $\E$ (see below for the definition). 

As the deterministic problem is out of reach mathematically, it has
been proposed an approach that models the chaotic effects of the
non-linearities by stochastic perturbations of the dynamics
that conserves energy. In the harmonic case, random exchanges of momentum of nearest
neighbor particles that conserve total energy but not momentum 
have been studied (\cite{bo, fno, ber}). 
Stochastic exchanges that also conserve total momentum have been
considered in \cite{BBO1, BBO2}, where a divergence of the diffusivity
is proven for unpinned harmonic chains. The stochastic perturbations
considered in these papers are very degenerate (of hypoelliptic type),
since they act only on the momenta of the particles, and not on the
positions. In particular these stochastic dynamics conserve also the
total length $\sum_j r_j$.

In this article we want to deal with anharmonic chains with noise that
conserves only energy. For reasons we will explain in a moment, we
need more elliptic stochastic perturbations that act also on the
positions. In the case of one-dimensional unpinned chains, there is a
way to define these perturbations locally (see the next section for
details) just using 
squares of vector fields that appear in the Liouville vector field that
generates the Hamiltonian dynamics. With these perturbations,
we have a dynamics that conserves \emph{only} the total energy. 
As a result, 
the dynamics has a one-parameter family of invariant measures given by
the grand canonical measures $\{\nu_{\beta}, \beta >0\}$ defined by 
\begin{align*}
\nu_{\beta}=\prod_j \frac{e^{-\beta \E_i}}{\sqrt{2\pi
    \beta^{-1}}Z_{\beta}} dp_j dr_j. 
\end{align*}
Notice that $\{r_j, p_j\}_j$ are independently
distributed under these probability measures. 
  
So we can consider the system starting with the equilibrium distribution
$\nu_{\beta}$ at temperature $T = \beta^{-1}$.  
We can prove the diffusive scaling limit by the results in this article in the following
linearized sense: 
define the space time energy covariance in equilibrium at temperature
$\beta^{-1}$:
\begin{equation*}
  C(i,j,t) = \mathbb E\left( \E_i(t) \E_j(0) \right) - \bar\E_\beta^2
\end{equation*}
where $\mathbb E$ denotes the expectation for the stochastic dynamics
starting with the grand-canonical measure at temperature $\beta^{-1}$
and $\bar\E_\beta$ is the expectation value of $\E_0$ under $\nu_{\beta}$.
In the following we will
denote simply by 
$e$ the corresponding value. 
Clearly $C(i,j,0) = \delta_{i,j} \chi_\beta$, where $\chi_\beta$ is the
variance of $\E_0$ under $\nu_\beta$. 
 Then it follows by our results that
 \begin{equation*}
   \lim_{N\to\infty} N C([Ny], [Nx], N^2 t) \ =\  
\frac{\chi_\beta}{\sqrt{2\pi D t}} e^{-(x-y)^2/2\mathcal D t}
 \end{equation*}
weakly, in the sense of the convergence of $N^{-1} \sum_{i,j} G(i/N)
F(j/N) C(i,j, N^2 t)$ for good test functions $G$ and $F$ of $\mathbb R$.
Here 
$ D = D(\bar\E_\beta)$ which is formally given by
\begin{equation}
  \label{eq:17}
  D\ = \ \lim_{t\to\infty} \frac 1t \sum_{i\in \mathbb Z} \ i^2\  C(i,0,t).
\end{equation}

We actually prove a stronger result at the level of fluctuation
fields. 
If the system is in equilibrium at temperature $\beta^{-1}$, then standard central
limit theorem for independent variables tells us that as $N\to\infty$
energy has Gaussian fluctuations, i.e. the energy fluctuation field
\begin{equation*}
  \label{eq:20}
  Y^N = \frac{1}{\sqrt N} \sum_i \delta_{i/N} \left\{\E_i(0) -
    \bar\E_\beta  \right\}
\end{equation*}
converges in law to a delta correlated centered Gaussian field $Y$
\begin{equation*}
   \mathbb E\left[ Y(F) Y(G) \right] = \chi_\beta \int
   F(y) G(y) dy. 
\end{equation*}
 
In this article we prove that these \emph{macroscopic} energy
fluctuations evolve diffusively in time (after a diffusive space-time
scaling) , i.e. that the time dependent distribution
\begin{equation*}
  \label{eq:22}
   Y^N_t = \frac{1}{\sqrt N} \sum_i \delta_{i/N} \left\{\E_i(N^2 t) -
     \bar\E_\beta\right\} 
\end{equation*}
converges in law to the solution of the linear SPDE
\begin{equation*}
  \label{eq:23}
  \partial_t Y = D\; \partial^2_y Y \; dt +
 \sqrt{2D \chi_\beta} \ \partial_y B(y,t)
\end{equation*}
where $B$ is the standard normalized space-time white noise.
In this sense, in equilibrium, energy fluctuation evolves
macroscopically following the linearized
heat equation.

The main point in the proof of this result is the following. Since
total energy is conserved, locally the energy of each particle is
changed by the energy currents with its neighbors, i.e. applying the
generator $L$ of the process to the energy $\E_i$ we obtain
\begin{equation}
  \label{eq:24}
  L \E_i = W_{i-1, i} - W_{i,i+1} 
\end{equation}
 where $W_{i,i+1} = -p_i V'(r_{i+1}) + W^S_{i,i+1}$. Here $-p_i
 V'(r_{i+1})$ is the instantaneous energy current associated to the
 Hamiltonian mechanism, while $W^S_{i,i+1}$ is the instantaneous
 energy current due to the stochastic part of the dynamics.
While \eqref{eq:24} provides automatically one space derivative
already at the microscopic level, $W_{i,i+1}$ is not a space-gradient.
In this sense this model falls in the class of the \emph{non-gradient
models}. Some of these non-gradient models have been studied with a
method introduced by Varadhan \cite{V}. 
The main point of this method is to prove that $W_{i,i+1}$ can be
approximated by a fluctuation-dissipation decomposition
\begin{equation*}
  \label{eq:25}
  W_{i,i+1} \ \sim \ D \nabla \E_i \ + \ L F
\end{equation*}
for a properly chosen sequence of local functions $F$. 
In the harmonic case of our model,
this decomposition
is exact for every configuration, i.e. there exists a local second
order polynomial $F$ such that $ W_{i,i+1}= D \nabla \E_i + L F$ for a
constant $D$ (cf. Remark \ref{rmk-harmonic} that contains an
equivalent decomposition). In the anharmonic case, such
decomposition can be only approximated by a sequence of local functions
$F_K$ in the sense that the difference has a small space-time
variance with respect to the dynamics in equilibrium at given temperature
(consequently $D$ is a function of this temperature). 

In order to do such decomposition, we have to use Varadhan's
approach to non-gradient systems \cite{V} and the generalization
to non-reversible dynamics \cite{X,K}. The main ingredients of the
methods are a \emph{spectral gap} for the stochastic part of the
dynamics, and a \emph{sector condition} for the generator $L$ of the
dynamics. It is in order to prove these properties that we need such
elliptic noise acting also on the positions.

We have to limit ourselves to these results on the equilibrium
fluctuation and we are not able to prove the full non-linear equation 
\eqref{eq:nonl} starting from a global non-stationary
profile. Unfortunately most of known techniques  to prove such
hydrodynamic limits in diffusive scaling are based on relative entropy
techniques (cf. \cite{GPV}, \cite{Y}, \cite{KL}) that do not work for
the energy diffusion in this model. 

The article is organized as follows: In Section \ref{model} we
introduce our model and state main results. In Section
\ref{sec-strategy}, we give
the strategy for proving the convergence of the finite dimensional
distributions.
 The complete proof is divided into several sections
\ref{sec:space-time-variance-bounds},
\ref{sec-hilbert}, \ref{sec-bgp}, \ref{sec:closedformch2} and \ref{sec:lie-algebra},
with sector condition proved in Section \ref{sec-sec}, and the spectral
gap in Section \ref{sec:spectral-gap-bound}.
The tightness shown in Section \ref{sec:tightness} concludes the proof.
Variational expressions for the thermal diffusivity are obtained in Section \ref{sec-diff-coeff} and some bounds on it are proven in Section
\ref{sec:upperboundch3}.

\section{Model and results}
\label{model}

We will now give a precise description of the dynamics. 
We consider a system of $N$ anharmonic oscillators in one-dimensional
space, whose 
hamiltonian dynamics is perturbed by a random dynamics that conserves
total energy. We consider a periodic boundary condition, but the
results can be generalized to different boundary conditions or also to
the infinite system.

Let $\T := (0, 1]$ be the $1$-dimensional torus, and for a positive
integer $N$ denote by $\T_N$ the lattice torus of length $N$ : $\T_N =
\{1, \ldots N\}$. The configuration space 
is denoted by $\Omega^N= (\R^2)^{\T_N}$ and a typical configuration is denoted by $\omega = (p_i,r_i)_{i \in \T_N}$
where $r_i= q_i - q_{i-1}$ represents the inter-particle distance
between the particles $i-1$ and $i$ (here we assume $a=0$ without loss of generality), and $p_i$ represents the velocity
of the particle $i$. All particles have mass equal to 1. 
The configuration evolves in time as a Markov process on $\R^{2N}$
with infinitesimal generator given by
\begin{displaymath}
L_N^{\gamma} = A_N+ \gamma S_N
\end{displaymath}
where 
\begin{displaymath}
A_N=  \sum_{i \in \T_N} (X_i-Y_{i,i+1}),  
\quad
S_N = \frac{1}{2} \sum_{i \in \T_N} \{(X_i)^2+(Y_{i,i+1})^2\}, 
\end{displaymath}
\begin{equation*}
  \label{eq:1}
  Y_{i,j} = p_i\partial_{r_j} - V'(r_j) \partial_{p_i}, \qquad X_i = Y_{i,i},
\end{equation*}
and $N+1 \equiv 1$. Notice that $A_N$ is the generator of the Hamiltonian dynamics (the
Liouville operator) while $S_N$ is the generator of the stochastic
perturbation. Here $\gamma>0$ is the strength of the stochastic
perturbation. We do not need any condition on $\gamma$, as
long as it is strictly positive. 

We assume that the function $V:\R \to \R_+$ satisfies the following
three properties: 
\begin{enumerate}[(i)]
\item  $V(r)$ is a smooth symmetric
  function.
\item $0 < \delta_{-} \le V''(r) \le \delta_+ < + \infty$.
\item \label{tech-ass} $\delta_-/\delta_+ > (3/4)^{1/16}$.
\end{enumerate}

\begin{remark}
The assumption (\ref{tech-ass}) is quite technical and required only in the proof of the spectral gap estimate in Section \ref{sec:spectral-gap-bound}.
\end{remark}

We denote the energy associated to the particle $i$ by 
$$\E_i = \frac {p_i^2}{2} + V(r_i)
$$
and the total energy defined by 
$H = \sum_{i \in \T_N} \E_i$ 
which denotes the Hamiltonian of the original deterministic dynamics.

\begin{remark}
  The total energy satisfies $L_N^{\gamma}(H) = 0$,
  i.e. total energy is a conserved quantity.
\end{remark}
\begin{remark}
  The other important conservation laws of the Hamiltonian dynamics, for the
  total length $\sum_{i \in \T_N} r_i$ and the total momentum $\sum_{i \in \T_N} p_i$, are
  destroyed by the stochastic noise $S_N$. In fact 
  $L_N (\sum_i r_i )= \gamma S_N (\sum_i r_i) = - \gamma\sum_i V'(r_i)$, 
 and   $L_N (\sum_i p_i) = \gamma S_N (\sum_i p_i) = -
 \frac{\gamma}{2}\sum_i (p_{i-1} +  p_i)V''(r_i)$. 
\end{remark}

We define a probability measure $\nu^N_{\beta}$ on $\Omega^N$ by
\begin{align*}
\nu^N_{\beta}(dpdr)=\prod_{i=1}^{N}\frac{\exp \left(-\beta \left(\frac{p_i^2}{2}+V \left(r_i \right) \right) \right)}{\sqrt{2\pi \beta^{-1}}Z_{\beta}}dp_i dr_i
\end{align*}
where
\begin{displaymath}
Z_{\beta}:=\int_{\R}e^{-\beta V(r)}dr < \infty.
\end{displaymath}

Denote by $L^2(\nu^N_{\beta})$ the Hilbert space of functions $f$ on $\Omega^N$ such that $\nu^N_{\beta}(f^2) < \infty$.
$S_N$ is formally symmetric on $L^2(\nu^N_{\beta})$ and $A_N$ is
formally antisymmetric on $L^2(\nu^N_{\beta})$. In fact, it is easy to
see that for smooth functions f and g in a core of the operators $S_N$
and $A_N$, we have for all $\beta > 0$ 
\begin{displaymath}
\int_{\R^{2N}} S_N(f) g \ \nu^N_{\beta}(dpdr) =\int_{\R^{2N}} f S_N(g) \
\nu^N_{\beta}(dpdr), 
\end{displaymath}
and
\begin{displaymath}
\int_{\R^{2N}} A_N(f)g \  \nu^N_{\beta}(dpdr) =-\int_{\R^{2N}} f A_N(g) \
\nu^N_{\beta}(dpdr). 
\end{displaymath}
In particular, the diffusion is invariant with respect to all the measures $\nu^N_{\beta}$. The distribution $\nu_{\beta}^N$ is called canonical
Gibbs measure at temperature $T =\beta^{-1}$. Notice that ${r_1,
  . . . , r_N, p_1, . . . , p_N}$ are independently distributed under
this probability measure. 

On the other hand, for every $\beta > 0$ the Dirichlet form of the diffusion with respect to $\nu^N_{\beta}$ is given by
\begin{displaymath}
\D_{N, \beta}(f)= \frac{\gamma}{2} \int_{\R^{2N}}  \sum_{i \in \T_N} \{ [X_i(f)]^2+[Y_{i,i+1}(f)]^2 \} \nu^N_{\beta}(dpdr).
\end{displaymath}

We will use the notation $\nu_\beta$ for the product measures on the configuration space $\Omega:=(\R^2)^{\Z}$, namely on the infinite lattice with marginal given by $\nu_\beta |_ {\{1,2, \dots, N\}} = \nu^N_{\beta}$. The expectation with respect to $\nu_\beta$ will be sometimes denoted by 
\begin{displaymath}
\int_{\Omega}  f \nu_{\beta}(dpdr) =  \langle f \rangle _{\beta}.
\end{displaymath}

Denote by $\{\omega(t)=(p(t),r(t)); t \ge 0 \}$ the Markov process
generated by $N^2L_N^{\gamma}$ (the factor $N^2$ 
corresponds to an acceleration of time). Let $C (\R_+, \Omega^N)$ be
the space of continuous 
trajectories on the configuration space. For any fixed time $T > 0$ and for a given measure $\mu^N$ on $\Omega^N$, the probability measure on $C([0, T], \Omega^N)$ induced by this Markov process starting from $\mu^N$ will be denoted by $\mathbb{P}_{\mu^N}$ . As usual, expectation with respect to  $\mathbb{P}_{\mu^N}$ will be denoted by $\mathbb{E}_{\mu^N}$.
The diffusion generated by $N^2 L_N^{\gamma}$ can also be described by
the following system 
of stochastic differential equations
\begin{align*}
dp_i(t) =  N^2 [V'(r_{i+1} )-V'(r_i)-\frac{\gamma p_i}{2}& \{V''(r_i)+V''(r_{i+1}) \} ] dt  \\
&- \sqrt{\gamma}N \{V'(r_{i+1})dB_i^1 + V'(r_i)dB_i^2 \}, \\
dr_i(t) = N^2 [p_i- p_{i-1}- \gamma V'(r_i) ] dt  + & \sqrt{\gamma}N
\{p_{i-1}dB_{i-1}^1 + p_i dB_i^2 \}   
\end{align*}
where $ \{B_i^1, B_i^2\}_{i \in \T_N}$ are $2N$-independent standard
Brownian motions.

Since total energy is conserved, the movement is constrained on the
\emph{microcanonical} surface of constant energy
\begin{equation}
  \label{eq:mcsurf}
  \Sigma_{N,E} = \left\{\omega \in \Omega^N ; \sum_{i=1}^N \E_i = N E \right\}.
\end{equation}
Our conditions on $V$ assure that these surfaces are always
connected. 
The vector fields $\{X_i, Y_{i,i+1}, i=1, \dots, N\}$ are tangent to
this surface, and as we show in Section \ref{sec:lie-algebra}, $\text{Lie}\{X_i, Y_{i,i+1}, i=1,
\dots, N\}$ generates the all tangent space. Consequently the
\emph{microcanonical} measures 
$$
\nu_{N,E}(\cdot) = \nu^N_{\beta}(\cdot|\Sigma_{N,E})
$$
are ergodic for our dynamics. We could have chosen $\nu_{N,E}$ as
initial distribution, but since by the equivalence of ensembles it
converges to $\nu_{\beta(E)}$ as $N\to\infty$, it would have been
irrelevant. Here $\beta(E)$ is defined as the inverse function of
\begin{equation}
  \label{eq:4}
  \bar\E_\beta := E(\beta) = \int \E_0 d\nu_\beta =\frac 1{2\beta} +
  \langle V(r_0) \rangle _{\beta}\ .
\end{equation}

By Ito's formula, we have
\begin{equation}
d\E_i(t) =  N^2 [W_{i-1,i}-W_{i,i+1}] dt  + N \{ \sigma_{i-1,i} dB_{i-1}^1 - \sigma_{i,i+1} dB_i^1 \}\label{eq:3}
\end{equation}
where 
\begin{equation}
  \begin{split}
    W_{i,i+1} &= W^A_{i,i+1}+W^S_{i,i+1},\\
    W^A_{i,i+1} &= -p_iV'(r_{i+1}),\\
    W^S_{i,i+1} &= \frac{\gamma}{2}\{p_i^2V''(r_{i+1})-V'(r_{i+1})^2 \},\\
    \sigma_{i,i+1} &= \sqrt{\gamma}p_iV'(r_{i+1}).\label{eq:2}
  \end{split}
\end{equation}

We can think of $W_{i,i+1}$ as being the instantaneous microscopic
current of energy 
between the oscillator $i$ and the oscillator $i + 1$. 
Observe that the current $W_{i,i+1}$ cannot be written as the
gradient of a local function, neither by an exact
fluctuation-dissipation equation, 
i.e. as the sum of a gradient and a dissipative term of the form
$L_N^{\gamma}(\tau_i h)$. That means, we are in the nongradient case.

Define the 
empirical energy distribution associated to the process by
\begin{equation*}\label{eq:empiricalch3}
\pi^N_t( \omega, du)=\frac{1}{N} \sum_{i \in \T_N} \E_{i}(t) \de_{\frac{i}{N}}(du),\quad  0 \le t \le T, \quad u \in \T,
\end{equation*}
and $ \langle \pi^N_t, f \rangle $ stands for the integration of $f$
with respect to $\pi^N_t$.

Notice that we are using here as space variable the \emph{material}
coordinate $i/N$, and not the physical positions $q_i$.
These two descriptions are equivalent but in our model the
material (Lagrangian) coordinates simplify notations.

It is easy to prove that, starting with the equilibrium measure
$\nu^N_{\beta}$ (or with the microcanonical $\nu_{N,\bar\E_\beta}$), we have 
$\pi^N_t \longrightarrow \bar\E_\beta du$
as weak convergence in probability.

We want to investigate the fluctuation of the empirical measure
$\pi^N$ with respect to this limit.
 Denote by $Y^N_t$ the \emph{empirical energy fluctuation field} acting on a smooth function $H : \T \to \R$ as
\begin{displaymath}
Y^N_t (H) = \frac{1}{\sqrt{N}} \sum_{i \in \T_N} H(\frac{i}{N}) \{\E_{i}(t) - \bar\E_\beta\}.
\end{displaymath}
The limit process will be described by $\{Y_t\}_{ t \ge 0}$, the
stationary generalized Ornstein-Uhlenbeck process with zero mean and
covariances given by 
\begin{displaymath}
\mathbb{E}[Y_t(H_1)Y_0(H_2)] = 
\frac{\chi_\beta}{\sqrt{4\pi tD_\beta}}\iint_{\R^2}du dv
\bar{H_1}(u) e^{-\frac{(u-v)^2}{4 tD_\beta}} \bar{H_2}(v) 
\end{displaymath}
for every $t \ge 0$. Here $\chi_\beta$ stands for the variance
of the energy (the thermal capacity in this context) given by 
$$
\chi_\beta=\langle \E_0^2 \rangle _{\beta}- \langle \E_0 \rangle
_{\beta} ^2 = \frac{1}{2 \beta^2}- \frac d{d\beta} <V(r_0)>_\beta
$$ 
and $\bar{H_1} (u)$ (resp. $\bar{H_2} (u)$) is the periodic
extension of the smooth function $H_1$ (resp. $H_2$) to the real line,
and $D_\beta$ is the diffusion coefficient determined later 
(see Corollary \ref{cor:decompositionch3} and formula
\eqref{eq:variationalch3}).  

Consider for $k > \frac{5}{2}$ the Sobolev space
$\mathfrak{H}_{-k}$ of the distributions $Y$ on $\mathbb T$ such that
they have finite norm
\begin{equation*}
  \|Y\|_{-k}^2 = \sum_{n\ge 1} (\pi n)^{-2k} |Y(e_n)|^2
\end{equation*}
with $e_n(x) = \sqrt 2 \sin(\pi n x)$. Here $Y(e_n)$ denotes the
distribution $Y$ applied to the function $e_n(x)$.

Denote by $\mathbb{Q}_N$ the probability measure
on $C([0,T] , \mathfrak{H}_{-k})$ induced by the energy fluctuation
field $Y^N_t$ and the Markov process $\{ \omega^N(t), t \ge 0\}$
defined at the beginning of this section, starting from the
equilibrium probability measure $\nu^N_{\beta}$. Let $\mathbb{Q}$ be
the probability measure on the space $C([0, T], \mathfrak{H}_{-k})$
corresponding to the generalized Ornstein-Uhlenbeck process $Y_t$
defined above. 
We are now ready to state the main result of this work.

\begin{theorem} \label{thm:main}
The sequence of the probability measures $\{\mathbb{Q}_N\}_{N \ge 1}$ converges weakly to the probability measure $\mathbb{Q}$.
\end{theorem}

\begin{remark}
For each $H \in C^{\infty}(\T)$, 
\begin{align}
M^{H}_t:=Y_t (H) -Y_0(H) - \int^t_0 D_\beta Y_s(H'')ds  \label{eq:martingaleproblem1}, 
\end{align}
and
\begin{align}
N^{H}_t:=(M^{H}_t)^2 - 2t\chi_\beta D_\beta \int_\T H'(u)^2
du \label{eq:martingaleproblem2} 
\end{align}
are $L^1(\mathbb{Q})$-martingales.
\end{remark}

\section{Strategy of the proof of the main theorem}
\label{sec-strategy}

 We follow the argument in Section 11 in \cite{KL}. Theorem \ref{thm:main} follows from the following three statements:
\begin{enumerate}[(i)]
\item \label{tight} $\{ \mathbb{Q}_N\}_{N \ge 1}$ is tight,
\item  \label{clt} the one--time marginal of any limit point
  $\mathbb{Q}^*$ of a convergent subsequence of $\{ \mathbb{Q}_N\}_{N
    \ge 1}$ is the law of a centered Gaussian field $Y$ with
  covariance given by  
\begin{displaymath}
\mathbb{E}[Y(H_1)Y(H_2)] = \chi_\beta \int_\T H_1(u) H_2(u) du,
\end{displaymath}
\item \label{limp} all limit points $\mathbb{Q}^*$ of convergent subsequences of $\{ \mathbb{Q}_N\}_{N \ge 1}$ solve the martingale problems (\ref{eq:martingaleproblem1}) and (\ref{eq:martingaleproblem2}).
\end{enumerate}

The proof of (\ref{clt}) is obtained by a direct consequence of the central limit
theorem for independent variables. We will prove (\ref{tight}) in
section \ref{sec:tightness}.
We prove here the main point, i.e. (\ref{limp}).

For a given smooth function $H : \T \to \R$, we begin
by rewriting $Y^N_t(H)$ as 
\begin{equation}\label{eq-evo}
Y^N_t (H)=  Y^N_0 (H) + \int^t_0 \sqrt{N}\sum_{i \in \T_N}\nabla^NH
(\frac{i}{N})W_{i,i+1} ds + M^{H,N}(t) 
\end{equation}
where $\nabla^NH$ represents the discrete derivative of $H$:
\begin{displaymath}
(\nabla^NH)(\frac{i}{N})=N[H(\frac{i+1}{N})-H(\frac{i}{N})]
\end{displaymath}
and the martingale $M^{H,N}(t)$ is
\begin{displaymath}
M^{H,N}(t) = \int^t_0 \frac{1}{\sqrt{N}}\sum_{i \in \T_N}\nabla^NH (\frac{i}{N})\sigma_{i,i+1} dB^1_{i}.
\end{displaymath}

Denote by $\Cy$ the set of all smooth local functions $F$ on
$\Omega=(\R^2)^{\Z}$ with compact support and zero average with
respect to $\nu_\beta$, and
define the formal sum 
\begin{equation}
  \label{eq:5}
  \Gamma_F = \sum_{j\in\Z} \tau_j F
\end{equation}
where $\tau_j$ is the shift on $\Z$. 
Observe that, since $F$ is local, $X_{i} \Gamma_F$ and $Y_{i,i+1} \Gamma_F$
are always well-defined functions.

We want to introduce here a fluctuation-dissipation approximation of
the current $W_{i,i+1}$ in $\kappa_\beta (p^2_{i,i+1} - p^2_{i}) +
L^{\gamma} \tau_i F$, for a proper constant $\kappa_\beta
:=D_\beta\chi_\beta\beta^2$, and for that  
purpose we can decompose \eqref{eq-evo} with any fixed $F \in \Cy$
as follows:  
\begin{align}
\label{eq:decomp}
Y^N_t(H) =  Y^N_0(H) & + \int^t_0 Y^N_s ( D_\beta \Delta^NH) ds +
M^1_{N,F,t}(H) \\
\nonumber
& + I^1_{N,F,t}  +  I^2_{N,F,t}  + M^2_{N,F,t}
+ \kappa_\beta I^3_{N,t} 
\end{align}
where
\begin{align*}
&I^1_{N,F,t} (H) = \int^t_0 \sqrt{N}\sum_{i \in \T_N}\nabla^NH
(\frac{i}{N})[W_{i,i+1} + \kappa_\beta (p_{i+1}^2-p_i^2
) - L_N^{\gamma} (\tau_i F)]ds, \\ 
& I^2_{N,F,t} (H) = \int^t_0 \sqrt{N}\sum_{i \in \T_N}\nabla^NH
(\frac{i}{N}) L_N^{\gamma} (\tau_i F)ds, \\ 
& I^3_{N,t} (H) =   \int^t_0 \frac{1}{\sqrt{N}}\sum_{i \in \T_N}\Delta^NH (\frac{i}{N})[ (p_i^2 - \frac{1}{\beta}) -\frac{1}{\chi_\beta \beta^2} \{ \E_i - \bar\E_\beta \} ],  \\
& M^1_{N,F,t}(H) = \int^t_0 \frac{1}{\sqrt{N}}\sum_{i \in \T_N}\nabla^NH (\frac{i}{N})[ (\sigma_{i,i+1} - \sqrt{\gamma} Y_{i,i+1}( \Gamma_F) )dB^1_{i}  - \sqrt{\gamma} X_{i} (\Gamma_F) dB^2_{i}] , \\
& M^2_{N,F,t}(H) = \int^t_0 \frac{\sqrt{\gamma}}{\sqrt{N}}\sum_{i \in
  \T_N}\nabla^N H (\frac{i}{N})[ Y_{i,i+1} (\Gamma_F) dB^1_{i}  +
X_{i} (\Gamma_F) dB^2_{i}] . 
\end{align*}

The proof of (\ref{limp}) is reduced to the following lemmas:

\begin{lemma}\label{lem:rest}
For every smooth function $H: \T \to \R$ and every function $F \in \Cy$,
\begin{equation*}
\lim_{N \to \infty} \mathbb{E}_{\nu^N_{\beta}}\left[ 
\sup_{ 0 \le t \le T}(I^2_{N,F,t} (H)+M^2_{N,F,t}(H))^2
\right]=0.
\end{equation*}
\end{lemma}

\begin{lemma}\label{lem:Boltzmann-Gibbs}
For every smooth function $H: \T \to \R$,
\begin{equation*}
\lim_{N \to \infty} \mathbb{E}_{\nu^N_{\beta}}\left[\sup_{0\le t\le T}
(I^3_{N,t} (H))^2\right] = 0.
\end{equation*}
\end{lemma}

\begin{lemma}\label{lem:CLT}
There exists a sequence of functions 
$\{F_{K}\}_{K \in \N} \in \mathcal C$ such that, 
for every smooth function $H: \T \to \R$,
\begin{equation*}
\lim_{K \to \infty} \lim_{N \to \infty}
\mathbb{E}_{\nu^N_{\beta}}\left[\sup_{0\le t\le T} (I^1_{N,F_K, t}
  (H))^2\right]=0.  
\end{equation*}
Moreover, for this sequence $\{F_{K}\}_{K \in \N}$, 
\begin{equation*}
\lim_{K \to \infty} E_{\nu_{\beta}}[(\sigma_{0,1} -
\sqrt{\gamma}Y_{0,1}( \Gamma_{F_K}) )^2  + (\sqrt{\gamma}X_0
(\Gamma_{F_K}) )^2]
= 2D_\beta\chi_\beta=\frac{2\kappa_\beta}{\beta^2}
\end{equation*}
\end{lemma}
Note that 
\[
I^1_{N,F,t} (H) = \int^t_0 \sqrt{N}\sum_{i \in \T_N}\nabla^NH (\frac{i}{N})[W_{i,i+1} + \kappa_\beta (p_{i+1}^2-p_i^2 ) - L_N (\tau_i F)] ds.
\]

As a consequence of Lemma \ref{lem:CLT}, the martingale
$M^1_{N,F_K,t}(H)$ will converge, as 
$N\to \infty$ and $K\to \infty$ to a martingale $M_t^{H}$ of quadratic
variation $2 t D_\beta\chi_\beta \int_\T H'(u)^2 du$, and the limit
$Y_t(H)$ of $Y^N_t(H) $ will satisfy the equation
\begin{equation}
  \label{eq:6}
  Y_t(H) = Y_0(H) + \int_0^t Y_s(D_\beta H'')\; ds + M_t^{H}.
\end{equation}
This martingale problem is solved uniquely by the generalize
Ornstein-Uhlenbeck process $Y_t$ defined above. 

Now we proceed to give a proof of Lemma \ref{lem:rest}.
Lemma \ref{lem:Boltzmann-Gibbs} will be proven in Section \ref{sec-bgp},
while Lemma \ref{lem:CLT} will be the content of the rest of the article.

\begin{proof}[Proof of Lemma \ref{lem:rest}]
Let us define
\begin{equation*}
\zeta_{N,F}(t) =\frac{1}{N^{\frac{3}{2}}} \sum_{i \in \T_N} \nabla^NH(\frac{i}{N}) \tau_iF(\omega^N_t).
\end{equation*}
From Ito's formula, we obtain
\begin{align*}
\zeta_{N,F}(t) &= \zeta_{N,F}(0) + I^2_{N,F,t} (H) \\
& + \int^t_0 \frac{\sqrt\gamma}{\sqrt N} 
\sum_{j \in \T_N} \Big[ Y_{j,j+1} \left(\sum_{i \in
  \T_N}\nabla^N H(\frac{i}{N}) \tau_i F\right)
dB^1_{j} \\
&\qquad\qquad\qquad + X_{j} \left(\sum_{i \in
  \T_N}\nabla^N H(\frac{i}{N}) \tau_i F\right) dB^2_{j} \Big] . 
\end{align*}
Therefore, 
\begin{align*}
& (I^2_{N,F,t} (H) +M^2_{N,F,t}(H))^2 \le   2 (\zeta_{N,F}(t) - \zeta_{N,F}(0))^2  \\
+  & 2 \Big\{  \int^t_0 \frac{\sqrt{\gamma}}{\sqrt{N}}\sum_{j \in
    \T_N} \Big[ Y_{j,j+1} \left(\sum_i \nabla^N H(\frac{i}{N}) \tau_i F
    -\nabla^N H(\frac{j}{N}) \Gamma_F \right) dB^1_{j}  \\
& + X_{j} \left(\sum_i \nabla^N H(\frac{i}{N}) \tau_i F
    -\nabla^N H(\frac{j}{N}) \Gamma_F \right) dB^2_{j}\Big] \Big\}^2  .
\end{align*}
Since $F$ is local and of null average and $H$ is smooth, it is easy
to see that the expectation of first
term is of order $\frac{1}{N^2}$. 
The expectation of the second term is equal to 
\begin{equation*}
  \begin{split}
    \frac{{2\gamma t}}{N} \sum_{j\in \T_N} \Big[ \left< Y_{j,j+1} \left(\sum_i \nabla^N H(\frac{i}{N}) \tau_i F
    -\nabla^N H(\frac{j}{N}) \Gamma_F \right)^2 \right>_\beta\\
  + \left< X_{j} \left(\sum_i \nabla^N H(\frac{i}{N}) \tau_i F
    -\nabla^N H(\frac{j}{N}) \Gamma_F \right)^2 \right>_\beta
    \Big]
  \end{split}
\end{equation*}

Since $F$ is local, it is easy to see that this is also of order
$\frac{1}{N^2}$. 
\end{proof}

The proof of Lemma \ref{lem:CLT} is the hard part of the paper. We
will need some tools to estimate space-time variances through
variational formulas, as explained in the next section. In order to
establish these variational formula we need some finite dimensional
approximations of the solutions, see section \ref{sec-sec}, where a
bound on the spectral gap of $S$ is
needed. This is proven in section \ref{sec:spectral-gap-bound}.

\section{Space-Time Variance bounds}
\label{sec:space-time-variance-bounds}

In the following we will simply denote by $\langle \cdot \rangle$ the
expectation with respect to the grand-canonical measure $\nu_\beta$.

In order to prove lemmas \ref{lem:Boltzmann-Gibbs} and \ref{lem:CLT},
we will make use of the following general bound for time variances: 

\begin{proposition}
  Let $F$ be a smooth function in $L^2(\nu_{\beta}^N)$ satisfying $E_{\nu_{N,E}}[F]=0$ for all $E>0$. Then
  \begin{equation}\label{eq-genbound}
    \mathbb E_{\nu_\beta^N} \left(\sup_{0\le t \le T}\left[  \int_0^t
      F(\omega_s^N) ds \right]^2\right) \ \le \ \frac{16 T}{\gamma N^2}
  \langle F  \left(-S_N\right)^{-1} F \rangle.
  \end{equation}
\end{proposition}
A proof of \eqref{eq-genbound} can be found in \cite{clo} or in
\cite{klo}. 

Observe that, by the spectral gap for $S_N$ proven in Section
\ref{sec:spectral-gap-bound},
 the right hand side is always well defined.
We want to use the bound \eqref{eq-genbound} for functions of the type
$F = \sum_j G(j/N) \tau_j \phi$, for a certain class of local
functions $\phi$.

First, we introduce some notations.  We denote $\tilde{\Cy}$ the set of smooth local functions $f$ on $\Omega=(\R^2)^{\Z}$ satisfying that
 \[
f \in L^2(\nu_{\beta}), \ X_i f(p,r)  \in L^2(\nu_{\beta}), \ Y_{i,i+1}f \in L^2(\nu_{\beta})
 \]
 for all $i \in \Z$. Note that $\Cy \subset \tilde{\Cy}$.
Here and after, we consider operators $L^{\gamma}$, $S^{\gamma}$ and
$A$ acting on functions $f$ in $\Cy$ as   
\[
L^{\gamma}f= S^{\gamma}f+Af, \quad S^{\gamma}f=\frac{\gamma}{2}\sum_{i \in \Z}\{(X_i)^2f+ (Y_{i,i+1})^2 f\}, \quad Af=\sum_{i \in \Z}X_i f - Y_{i,i+1}f .
\]
For a fixed positive integer $l$, we define $\La_l:=\{-l, -l+1,..., l-1,l\}$ and $L^{\gamma}_{\La_l}$, $S_{\La_l}^{\gamma}$ the restriction of the generator $L^{\gamma}$, $S^{\gamma}$ to $\La_l$ respectively. For $\Psi$ in $\Cy$, denote by $s_{\Psi}$ the smallest positive integer $s$ such that $\La_s$ contains the support of $\Psi$.
Let $\Cy_0$ be a subspace of local functions defined as follows:
\begin{displaymath}
\Cy_0 = \{\ f ;  f=\sum_{i \in  \Lambda}[ X_i (F_i) + Y_{i,i+1} (G_i) ] \ \text{for some} \ \Lambda \subset \subset \Z \text{ and} \{F_i\}_{i \in  \Lambda}, \{G_i \}_{i \in  \Lambda} \in \tilde{\Cy}  \}.
\end{displaymath}

First, we note some useful properties of the space $\Cy_0$.
\begin{lemma}\label{lem:cy0}
  \begin{enumerate}[(i)]
  \item For any $f \in \Cy_0$, $l \ge s_f+1$ and $E>0$,  $E_{\nu_{l,E}}[f]=0$.
\item $W^S_{0,1}$, $W^A_{0,1}$ and $p_1^2 - p_0^2$ are elements of
  $\Cy_0$.
\item For any $F \in \Cy$, $L^{\gamma}F$, $S^{\gamma}F$ and $AF$ are
  elements of $\Cy_0$.
  \end{enumerate}
\end{lemma}
\begin{proof}
(i) and (iii) are straightforward.

\noindent
(ii): We have
\begin{equation*}
  \begin{split}
    W^S_{0,1}  &=\frac{\gamma}{2} \{p_0^2V''(r_1)-V'(r_1)^2\}
    =\frac{\gamma}{2} Y_{0,1} (p_0V'(r_1))\\
    W^A_{0,1} &= -p_0V'(r_1)= Y_{0,1} (-V(r_1))\\
    p_1^2 - p_0^2 &= X_{1} \{(p_0+p_1) r_1\} - 
    Y_{0,1} \{(p_0+p_1) r_1\}.
  \end{split}
\end{equation*}
\end{proof}

Next, we study the variance
\begin{displaymath}
(2l)^{-1} \langle (-S_{\La_l}^{\gamma})^{-1}\sum_{|i| \le l_{\psi}}\tau_i \psi, \sum_{|i| \le l_{\psi}}\tau_i\psi \rangle
\end{displaymath}
for $\psi \in \Cy_0$ where $l_{\psi}= l-s_{\psi}-1$. 
We start with introducing a semi-norm on $\Cy$, which is closely
related to the central limit theorem variance. For cylinder functions
$g$, $h$ in  $\Cy$, let 
 \begin{equation}
\ll g,h \gg_{*}= \sum_{i \in \Z} \langle g, \tau_i h \rangle \quad
\text{and} \quad \ll g \gg_{**}=\sum_{i \in \Z}i \langle g,\E_i
\rangle.        \label{eq:7} 
\end{equation}
$\ll g,h \gg_{*}$ and $\ll g \gg_{**}$ are well defined because $g$
and $h$ are local and null average, therefore all but a finite number of
terms vanish.  

Notice that if $h \in \Cy_0$ with $h = \sum_{j\in\Lambda} X_j F_j +
Y_{j,j+1} G_j$ then we can compute
\begin{equation*}
  \begin{split}
    \ll h \gg_{**} &= \lim_{l\to\infty} - \sum_{j\in\Lambda}
    \sum_{i=-l}^l i \langle 
    G_j, Y_{j,j+1} \E_i \rangle \\
    &=  - \sum_{j\in\Lambda} \left(j \langle
    G_j, Y_{j,j+1} \E_j \rangle + (j+1)  \langle
    G_j, Y_{j,j+1} \E_{j+1} \rangle \right)\\
  &= -\sum_{j\in\Lambda} \langle p_j V'(r_{j+1}) G_j \rangle
  = - \langle p_0 V'(r_{1}) \sum_{j\in\Lambda} \tau_{-j} G_j \rangle .
  \end{split}
\end{equation*}
and for any $g\in \Cy$
\begin{equation*}
  \begin{split}
    \ll g, h \gg_{*} = -\sum_{j\in\Lambda} \left[ <F_j, X_j \Gamma_g> + <G_j,
    Y_{j,j+1} \Gamma_g> \right] \\
  =  - \left[ \left<\sum_{j\in\Lambda} \tau_{-j} F_j, X_0 \Gamma_g \right> +
    \left<\sum_{j\in\Lambda}\tau_{-j} G_j, Y_{0,1} \Gamma_g \right> \right]\ .
  \end{split}
\end{equation*}

For $h$ in $\Cy_0$, define the semi-norm 
$\triple h \triple_{-1}$ by  
\begin{align}
\triple h \triple_{-1}^2 =\sup_{g \in \Cy, a \in \R} & \{2\ll g,h
\gg_{*}+2a \ll h \gg_{**}  \nonumber \\ 
 & -\frac{\gamma}{2} \langle (a p_0V'(r_1)+ Y_{0,1}\Gamma_g)^2 \rangle
 - \frac{\gamma}{2} \langle (X_0\Gamma_g)^2 \rangle \}. \label{eq:seminormdef}
\end{align}
Observe that denoting $\tilde F = \sum_{j\in\Lambda} \tau_{-j} F_j,
\tilde G = \sum_{j\in\Lambda} \tau_{-j} G_j$, we can express this norm
as
\begin{align}
\triple h \triple_{-1}^2 =\sup_{g \in \Cy, a \in \R} & \{-2 <\tilde F,
X_0 \Gamma_g>  - 2 <\tilde G, Y_{0,1} \Gamma_g + a p_0
V'(r_1)>\nonumber  \\ 
 & -\frac{\gamma}{2} \langle (a p_0 V'(r_1)+ Y_{0,1}\Gamma_g)^2 \rangle
 - \frac{\gamma}{2} \langle (X_0\Gamma_g)^2 \rangle
 \}. \label{eq:seminormdef2} 
\end{align}
and clearly $\triple h \triple_{-1}^2 \le \frac 2\gamma (<\tilde F^2>
+ <\tilde G^2>)$.
  
We investigate several properties of the semi-norm $\triple \cdot
\triple_{-1}$ in the next section, while in this section we prove the
following 
key proposition:

\begin{proposition}\label{prop:variancech3}
Consider a local function $\psi$ in $\Cy_0$. 
Then,
\begin{displaymath}
\lim_{l \to \infty} (2l)^{-1} \langle
(-S_{\La_l}^{\gamma})^{-1}\sum_{|i| \le l_{\psi}}\tau_i \psi,
\sum_{|i| \le l_{\psi}}\tau_i \psi \rangle = \triple \psi
\triple_{-1}^2.  
\end{displaymath}
\end{proposition}

The proof is divided into two lemmas.
\begin{lemma}\label{lem:clt1}
For $\psi$ in $\Cy_0$
\begin{displaymath}
\liminf_{N \to \infty} (2l)^{-1} \langle
(-S_{\La_l}^{\gamma})^{-1}\sum_{|i| \le l_{\psi}}\tau_i \psi,
\sum_{|i| \le l_{\psi}}\tau_i \psi \rangle \ge  \triple \psi
\triple_{-1}^2.
\end{displaymath}
\end{lemma}

\begin{lemma}\label{lem:clt2}
For $\psi$ in $\Cy_0$
\begin{displaymath}
\limsup_{N \to \infty}(2l)^{-1} \langle
(-S_{\La_l}^{\gamma})^{-1}\sum_{|i| \le l_{\psi}}\tau_i \psi,
\sum_{|i| \le l_{\psi}}\tau_i \psi \rangle \le \triple \psi
\triple_{-1}^2. 
\end{displaymath}
\end{lemma}

\begin{proof}[Proof of Lemma \ref{lem:clt1}]
Define 
$$
A_l:= \sum_{ i =  -l}^{l-1}\tau_i W^S_{0,1}
$$ 
and for $F \in \Cy$,  let 
$$
H^F_l:= \sum_{|i| \le l-s_F-1}\tau_i S^{\gamma}F .
$$
It is easy to see that 
\begin{align}
 \lim_{N \to \infty} (2l)^{-1} \langle (- S^\gamma_{\La_l})^{-1}
 &\sum_{|i| \le l_{\psi}} \tau_i \psi, A_l \rangle = - \ll \psi
 \gg_{**}, \label{eq-var1}\\ 
 \lim_{l \to \infty} (2l)^{-1} \langle (-S_{\La_l}^{\gamma})^{-1}
 &\sum_{|i| \le l_{\psi}}\tau_i \psi, H^F_l \rangle = -
 \ll \psi,F \gg_{*}, \\   
 \lim_{l \to \infty} (2l)^{-1} \langle (- S^\gamma_{\La_l})^{-1} & (a
 A_l + H^F_l ) , a A_l + H^F_l \rangle  \nonumber \\ 
& = \frac{\gamma}{2} \langle (a p_0V'(r_1)+ Y_{0,1}\Gamma_F)^2 \rangle
+ \frac{\gamma}{2} \langle (X_0\Gamma_F)^2 \rangle.
\end{align}
We just prove here \eqref{eq-var1}, the other relations are proven in
similar way. Assume for the simplicity that $\psi = X_0 F + Y_{0,1}G$, the
general case follows by linearity.
Since $A_l = S^\gamma_{\La_l} \sum_{j=-l}^l j \E_j$
\begin{equation*}
  \begin{split}
     (2l)^{-1} \langle (- S^\gamma_{\La_l})^{-1}
     &\sum_{|i| \le l_{\psi}} \tau_i \psi, A_l \rangle
     = - (2l)^{-1}  \sum_{j=-l}^l  \sum_{|i| \le l_{\psi}} j \langle
     \psi, \E_{j-i} \rangle \\
     & =  (2l)^{-1} \sum_{|i| \le l_{\psi}} \sum_{j=-l}^{l}  j 
     \langle G, Y_{0,1} \E_{j-i} \rangle \\
     & =  (2l)^{-1} \sum_{|i| \le l_{\psi}} \left(i \langle G,
       Y_{0,1} \E_0\rangle + (i+1) \langle G,
       Y_{0,1} \E_1\rangle  \right)    \\
     & = (2l)^{-1} (2l_\psi +1) \langle G, p_0 V'(r_1) \rangle
  \mathop{\longrightarrow}_{l\to\infty} - \ll \psi \gg_{**} 
  \end{split}
\end{equation*}

Then, obviously, 
\begin{align*}
 \liminf_{l \to \infty} (2l)^{-1} &  \langle
 (-S_{\La_l}^{\gamma})^{-1}\sum_{|i| \le l_{\psi}} \tau_i \psi,
 \sum_{|i| \le l_{\psi}} \tau_i \psi \rangle \\
 \ge   \liminf_{l \to \infty}  & [ 2(2l)^{-1} \langle
 (-S_{\La_l}^{\gamma})^{-1}\sum_{|i| \le l_{\psi}} \tau_i \psi, -(a
 A_l + H^F_l) \rangle    \\
& - (2l)^{-1} \langle (-S_{\La_l}^{\gamma})^{-1} (a A_l + H^F_l), (a
A_l + H^F_l) \rangle  ]\\
= 2\ll \psi,F \gg_{*} & +2a \ll \psi \gg_{**} -\frac{\gamma}{2}
\langle (a p_0V'(r_1)+ Y_{0,1}\Gamma_F)^2 \rangle - \frac{\gamma}{2}
\langle (X_0\Gamma_F)^2 \rangle. 
\end{align*}
Then, taking the supremum of $a \in \R$ and $F \in \Cy$ we obtain the desired inequality.
\end{proof}

\begin{proof}[Proof of Lemma \ref{lem:clt2}]

Let us assume for simplicity of notation that $\psi = X_0 F + Y_{0,1}
G$. The general case will follow straightforwardly. 
We use the variational formula 
\begin{equation*}
  \begin{split}
    (2l)^{-1} & \langle (-S_{\La_l}^{\gamma})^{-1}\sum_{|i| \le
      l_{\psi}} \tau_i \psi, \sum_{|i| \le l_{\psi}} \tau_i \psi
    \rangle \\
    &= \sup_{h} \left\{2 \langle \psi , \frac{1}{2l}\sum_{|i| \le l_{\psi}}
    \tau_i h\rangle  - \frac{\gamma}{4l} \mathcal D_l(h)\right\} \\
    &= \sup_{h} \left\{2 \langle F
    X_0 (\frac{1}{2l}\sum_{|i| \le l_{\psi}}\tau_i h) + G
    Y_{0,1}(\frac{1}{2l} \sum_{|i| \le l_{\psi}}\tau_i h) \rangle  -
    \frac\gamma{4l} \mathcal D_l(h) \right\}
  \end{split}
\end{equation*}
where 
\[
\mathcal D_l(h) = \sum_{|i| \le l}  \langle  (X_i h)^2  \rangle +  \sum_{ i = -l}^{l-1} \langle (Y_{i,i+1}h)^2 \rangle.
\]
The supremum can be restrained in the class of functions $h$ that are
localized in $\Lambda_l$ and such that $ \mathcal D_l(h) \le C_{\psi} l$.

Notice that
\begin{equation*}
 \left| \langle F X_0 (\frac{1}{2l}\sum_{l_\psi\le |i| \le l}\tau_i
  h) + G Y_{0,1} (\frac{1}{2l}\sum_{l_\psi\le |i| \le l}\tau_i
  h)\rangle\right| \le \frac {C_{\psi}}{2l} \mathcal D_l(h)^{1/2}
\end{equation*}
so, calling 
\begin{equation*}
  \xi_0^l(h) = X_0 (\frac{1}{2l}\sum_{|i| \le l}\tau_i h), \qquad
  \xi_1^l(h) = Y_{0,1} (\frac{1}{2l}\sum_{|i| \le l}\tau_i h)
\end{equation*}
and observing that, by Schwarz inequality
\begin{equation*}
  \langle (\xi_0^l(h))^2 +(\xi_1^l(h))^2 \rangle \
 \le \frac 1{2l} \mathcal D_l(h)
\end{equation*}
we obtain the upper bound
\begin{equation*}
  \begin{split}
    (2l)^{-1}  \langle (-S_{\La_l}^{\gamma})^{-1}&\sum_{|i| \le
      l_{\psi}} \tau_i \psi, \sum_{|i| \le l_{\psi}} \tau_i \psi
    \rangle \\
\le \sup_h &\Big\{ 2 \langle F, \xi_0^l(h) \rangle + 2
      \langle G , \xi_1^l(h) \rangle 
        - \frac{\gamma}{2} \left(\langle (\xi_0^l(h))^2 +(\xi_1^l(h))^2
        \rangle\right)\Big\} + C_{\psi}l^{-1/2}.
    \end{split}
\end{equation*}
Since for any choice of a sequence $\{h_l\}_{l}$ satisfying $\mathcal D_l(h_l) \le C_{\psi} l$, we have that the sequence 
$(\xi_0^l(h_l),\xi_1^l(h_l))$ is uniformly bounded in $L^2(\nu_{\beta})$, we can extract
convergent subsequences in $L^2(\nu_{\beta})$.
  All limit vectors $(\xi_0, \xi_1)$ that we obtain as limit points of 
$(\xi_0^l(h_l),\xi_1^l(h_l))$ are \emph{closed} in the sense specified in
Section \ref{sec:closedformch2}. We call this set of closed functions
$\frak h_c$, and we have obtained that 
\begin{equation*}
  \begin{split}
    \limsup_{l\to\infty} (2l)^{-1} \langle
    (-S_{\La_l}^{\gamma})^{-1}&\sum_{|i| \le l_{\psi}} \tau_i \psi,
    \sum_{|i| \le l_{\psi}} \tau_i \psi \rangle \\\le
    &\sup_{(\xi_0,\xi_1)\in \frak h_c} \Big\{ 2 \langle F, \xi_0
    \rangle + 2 \langle G , \xi_1 \rangle -\frac{\gamma}{2} \left(\langle
      (\xi_0)^2 +(\xi_1)^2 \rangle\right)\Big\}
  \end{split}
\end{equation*}
and the desired upper bound follows from the characterization of $\frak
h_c$ proved by Theorem
\ref{thm:closedformch3} in Section \ref{sec:closedformch2}. More
specifically we prove there that a closed function $(\xi_0,\xi_1)$ can
be approximated in $L^2(\nu_\beta)$ by functions of the type
$(X_0\Gamma_g, Y_{0,1} \Gamma_g + a p_0 V'(r_1))$, with $g\in \Cy$ and
$a\in \R$.
\end{proof}  

We are now in the position to state the main result of this section:
\begin{theorem}\label{gbg}
  Let $\psi \in \Cy_0$, and $G$ a smooth function on $\T$. Then
  \begin{equation}
    \label{eq:8}
    \begin{split}
      \limsup_{N\to\infty} \mathbb E_{\nu_\beta^N} \left( \sup_{0\le
          t\le T} \left[ N^{1/2} \int_0^t \sum_{i\in \T_N} G(i/N)
          \tau_i \psi(\omega_s) ds \right]^2 \right) \\
      \le \frac{C
        T}{\gamma} \triple \psi \triple_{-1}^2 \int_\T G(u)^2 du .
    \end{split}
  \end{equation}
\end{theorem}

\begin{proof}[Proof of Theorem \ref{gbg}]
We follow the argument in \cite{clo}, Theorem 4.2. 

First we prove the simpler bound
\begin{equation}
  \label{eq:9-1}
  \mathbb E_{\nu_\beta^N} \left( 
      \sup_{0\le t\le T} \left[ N^{1/2} \int_0^t \sum_{i\in \T_N}
        G(i/N) \tau_i \psi(\omega_s) ds \right]^2 \right)
    \le \frac{C_\psi T}{\gamma} \frac 1N \sum_{i\in\T_N} G(i/N)^2
\end{equation}
for some finite constant $C_\psi$.

By \eqref{eq-genbound}, the left side of \eqref{eq:9-1} is bounded by
\begin{equation*}
  {16 T} \langle N^{1/2}\sum_{i\in \T_N}
        G(i/N) \tau_i \psi, (-N^2 \gamma S_N)^{-1}   N^{1/2} \sum_{i\in \T_N}
        G(i/N) \tau_i \psi \rangle
\end{equation*}
 that can be written with the variational formula
 \begin{equation*}
  {16 T}  \sup_f \left\{ N^{1/2}\sum_{i\in \T_N}
        G(i/N) \langle f \tau_i \psi\rangle - N^2\gamma \langle f, (-S_N)
        f\rangle \right\}.
 \end{equation*}
Since $\psi \in \Cy_0$, there exists $\Psi_j$ and $\Phi_j$ belonging
to $\Cy$, for $j\in \Lambda \subset \subset \Z$, such that 
$\psi = \sum_{j\in\Lambda} \left[ X_j(\Psi_j) + Y_{j,j+1}
  (\Phi_j)\right]$, and we can bound by integration by parts
\begin{equation*}
  \langle f \tau_i \psi\rangle \le C_{\psi} \langle \sum_{j \in \La} [(X_{-i-j}f) ^2 + (Y_{-i-j,-i-j+1}f) ^2] \rangle^{1/2}
\end{equation*}
and again by Schwarz inequality
\begin{equation*}
   N^{1/2}\sum_{i\in \T_N}
        G(i/N) \langle f \tau_i \psi\rangle \le 
        \left(\frac 1N \sum_{i\in\T_N} G(i/N)^2\right)^{1/2}
        \left(N^2C_{\psi} \langle f, (-S_N) f\rangle\right)^{1/2}
\end{equation*}
and maximizing on $f$ we obtain \eqref{eq:9-1}.
  
Now we have to refine the bound by showing that the constant on the
right hand side is proportional to $\triple \psi \triple_{-1}^2$.
In order to do this, we have to perform a further microscopic average:
given $K<< N$, in \eqref{eq:8} we want to substitute
\begin{equation*}
  \sqrt N \sum_{i\in \T_N} G(i/N) \tau_i \psi
\end{equation*}
with 
\begin{equation*}
  \sqrt N \sum_{j\in \T_N} G(j/N) \frac 1{2K+1} \sum_{|i-j|\le K} \tau_i \psi.
\end{equation*}
Then the difference is estimated by 
\begin{equation*}
   \mathbb E_{\nu_\beta^N} \left( 
      \sup_{0\le t\le T} \left[ N^{1/2} \int_0^t 
        \sum_{i,j\in \T_N, |i-j| \le K}
        \frac 1{2K+1} (G(i/N)- G(j/N)) \tau_i \psi(\omega_s) ds \right]^2 \right)
\end{equation*}
that by \eqref{eq:9-1} is bounded by $C_G K/N$ and tends to $0$ as $N\to
\infty$.  

So we are left with 
\begin{equation*}
   \mathbb E_{\nu_\beta^N} \left( 
      \sup_{0\le t\le T} \left[\frac{ N^{1/2}}{2K+1} \int_0^t 
        \sum_{j\in \T_N} G(j/N) \tau_j \hat\psi_K (\omega_s) ds \right]^2 \right) 
\end{equation*}
where $\hat\psi_K = \sum_{ |i| \le K} \tau_i \psi$. By
(\ref{eq-genbound}) this is bounded by
\begin{equation*}
  \begin{split}
    \frac {CT}{2K+1} \sum_{j\in \T_N} \sup_f \left\{ \sqrt N G(j/N)
      \langle \tau_j \hat\psi_K f \rangle - N^2 \frac{\gamma}{2} \sum_{|i-j|\le K}
      \langle (X_i f)^2 + (Y_{i,i+1} f)^2\rangle \right\}\\
    \le \frac {CT}{\gamma N}\sum_{j\in \T_N}  G(j/N)^2 \frac{1}{2K+1}
    \langle \hat\psi_K, (-S_{\Lambda_K})^{-1} \hat\psi_K \rangle .  
  \end{split}
\end{equation*}
Taking the limit as $N\to\infty$ and $K\to \infty$ we obtain (\ref{eq:8}). 
\end{proof}

Applying Theorem \ref{gbg} to $I^1_{N,F,t}(H)$ we have
\begin{equation}
  \label{eq:10}
  \begin{split}
    \limsup_{N\to\infty} \mathbb E_{\nu_\beta^N} &\left( \sup_{0\le
        t\le T}(I^1_{N,F,t}(H))^2 \right) \\
    &\le \frac{C T}{\gamma}
    \triple W_{0,1} + \kappa_\beta (p_1^2 - p_0^2) - L^\gamma F
    \triple_{-1}^2 \int_\T H'(u)^2 du .
  \end{split}
\end{equation}
To conclude the proof of Lemma \ref{lem:CLT}, we need to show that
there exists a sequence of local functions $F_K$ in $\Cy$ such that 
 $$
\triple W_{0,1} + \kappa_\beta (p_1^2 - p_0^2) - L^\gamma F_K \triple_{-1}
\to 0
$$ 
as $K\to \infty$. This will be proven in the next section.

\section{Hilbert space}
\label{sec-hilbert}

In this section, to prove the first statement of Lemma \ref{lem:CLT}, we investigate the properties of the semi norm $\triple \cdot
    \triple_{-1}$  introduced in the previous section and the
    structure of the Hilbert space that it generates. 
 
We first define from $\triple \cdot \triple_{-1}$ a semi-inner product
on $\Cy_0$ through polarization: 
\begin{equation}
\ll g,h \gg_{-1} = \frac{1}{4}\{\triple g+h \triple_{-1}^2 - \triple g-h \triple_{-1}^2 \}. \label{eq:semiinch3}
\end{equation}

It is easy to check that (\ref{eq:semiinch3}) defines a semi-inner
product on $\Cy_0$. Denote by $\mathcal{N}$ the kernel
of the semi-norm $\triple \cdot \triple_{-1}$ on $\Cy_0$. Since  
$\ll \cdot , \cdot \gg_{-1}$ is a semi-inner product on
$\Cy_0$, the completion of $\Cy_0|_{\mathcal{N}}$,
denoted by $\mathcal{H}_{-1}$, is a Hilbert space. 

In the following, in order to simplify notations, we will set $L=
L^\gamma$ and $S = S^\gamma$.
 
By Lemma \ref{lem:cy0}, the linear space generated by $W^S_{0,1}$ and
$S\Cy := \{ Sg; \ g \in \Cy \}$ are subsets of
$\Cy_0$. The first main result of this section consists in showing
that $\mathcal{H}_{-1}$ is the completion of
$S\Cy|_{\mathcal{N}}+\{ W^S_{0,1} \}$, in
other words, that all elements of $\mathcal{H}_{-1}$ can be
approximated by $aW^S_{0,1}+Sg$ for some $a$ in $\R$ and $g$
in $\Cy$. To prove this result we derive two elementary identities: 
\begin{equation}
\ll h,Sg \gg_{-1} = - \ll h,g \gg_{*}, \quad \ll h,W^S_{0,1}
\gg_{-1} = -\ll h \gg_{**} \label{eq:keyidch3}
\end{equation}
for all $h$ in $\Cy_0$ and $g$ in $\Cy$.

By Proposition \ref{prop:variancech3} and (\ref{eq:semiinch3}), the
semi-inner product $\ll h,g \gg_{-1}$ is the limit of the covariance
$(2N)^{-1} \langle (-S_{\La_N})^{-1}\sum_{|i| \le N_g}\tau_i g,
\sum_{|i| \le N_h}\tau_i h \rangle$ as $N \uparrow
\infty$
. 
In particular, if $g=Sg_0$, for some
$g_0$ in $\Cy$, the inverse of the operator $S$ cancels with the
operator $S$. Therefore
\begin{displaymath}
\ll h,Sg_0 \gg_{-1} = -\lim_{N \to \infty} (2N)^{-1} \langle \sum_{|i|
  \le N_{g_0}}\tau_i g_0, \sum_{|i| \le N_h}\tau_i h \rangle =
- \ll g_0,h \gg_{*}. 
\end{displaymath}
The second identity is proved by similar way with the elementary
relation 
 $$
S_{\La_N}\sum_{i \in \La_N}i \E_i=\sum_{i,i+1 \in \La_N}W^S_{i,i+1}.
$$ 

The identities of (\ref{eq:keyidch3}) permit to compute the following
elementary relations 
\begin{align*}
\ll W^S_{0,1}, Sh \gg_{-1} & =-\sum_{i \in\Z} i  \langle \E_i
S h \rangle 
 =  \gamma \langle p_0 V'(r_1) Y_{0,1} \Gamma_h \rangle, \\
\ll p_1^2-p_0^2, Sh \gg_{-1} &=0
\end{align*}
for all $h \in \Cy$, and
  \begin{equation*}
    \begin{split}
      \ll W^S_{0,1}, W^S_{0,1} \gg_{-1} &= \frac{\gamma}{2} \langle
      (p_0 V'(r_1) )^2 \rangle = \frac{\gamma}{2\beta} \langle
      V'(r_1)^2 \rangle, \\
      \ll W^S_{0,1}, p_1^2-p_0^2 \gg_{-1} &= -\frac{1}{\beta^2}.
    \end{split}
\end{equation*}
Furthermore,
\begin{displaymath}
\triple a W^S_{0,1} +Sg \triple_{-1}^2 = \frac{\gamma}{2}
\langle (a p_0V'(r_1)+ Y_{0,1}\Gamma_g)^2 \rangle +
\frac{\gamma}{2} \langle  ( X_{0}\Gamma_g )^2 \rangle    
\end{displaymath}
for $a$ in $\R$ and $g$ in $\Cy$. In particular, the variational
formula (\ref{eq:seminormdef}) for $\triple h \triple_{-1}^2$ is reduced to the expression 
\begin{equation}
\triple h \triple_{-1}^2 = \sup_{a \in \R, g \in \Cy}\{ -2 \ll h, a
W^S_{0,1} +Sg  \gg_{-1} - \triple a W^S_{0,1} +Sg
\triple_{-1}^2 \}. \label{eq:normvariationalch3} 
\end{equation}
\begin{proposition}\label{prop:directch3}
Recall that we denote by $S\Cy$ the space $\{ Sg; \ g \in \Cy
\}$. Then we have 
\begin{displaymath}
\mathcal{H}_{-1} = \overline{S\Cy}|_{\mathcal{N}} \plus \{
W^S_{0,1} \}. 
\end{displaymath}
\end{proposition}
\begin{proof}
The inclusion $\mathcal{H}_{-1} \supset
\overline{S\Cy}|_{\mathcal{N}} \plus \{ W^S_{0,1} \}$ is
obvious. Then we have only to show that if $h\in \mathcal{H}_{-1}$
such that 
$\ll h, W^S_{0,1}\gg =0$ and $\ll h, Sg \gg =0$ for all $g\in \Cy$,
then $\triple h\triple_{-1} = 0$. This follows directly from the
variational formula (\ref{eq:normvariationalch3}).
\end{proof}

\begin{corollary}\label{cor:directgradient}
We have
\begin{displaymath}
\mathcal{H}_{-1} =
\overline{S\Cy}|_{\mathcal{N}} \oplus \{ W_S^{0,1} \} =
\overline{S\Cy}|_{\mathcal{N}} \oplus \{ p_1^2 - p_0^2 \}. 
\end{displaymath}
\end{corollary}
\begin{proof}
Since $\ll p_1^2-p_0^2, Sh \gg_{-1}=0$ for all $h \in \Cy$ and
$\ll W^S_{0,1}, p_1^2-p_0^2 \gg_{-1} = -\frac{1}{\beta^2}$, the result follows from
Proposition \ref{prop:directch3} straightforwardly.
\end{proof}

\begin{remark}
While the statement of Proposition \ref{prop:directch3} claims that $\mathcal{H}_{-1}$ is generated by the spaces $S\Cy$ and $\{
W^S_{0,1} \} (=\{a W^S_{0,1}; a \in \R \})$, the statement of Corollary \ref{cor:directgradient} claims also that the intersection of them is the trivial set.
Note that $S\Cy$ and $\{W^S_{0,1} \}$ are not orthogonal.

\end{remark}

Next, to replace the space $S\Cy$ by $L\Cy$, we show some useful lemmas.
\begin{lemma}\label{lem:antich3}
For all $g,h \in \Cy$, $\ll Sg, A h \gg_{-1} = -\ll
Ag, Sh \gg_{-1}$.  Especially, $\ll Sg, A
g \gg_{-1}=0$. 
\end{lemma}
\begin{proof}
By the first identity of (\ref{eq:keyidch3}),
\begin{align*}
\ll Sg, Ah \gg_{-1}&=-\ll g, Ah \gg_{*}=-\sum_{i \in \Z}\langle  \tau_i g, Ah \rangle \\
 &=\sum_{i \in \Z}\langle A \tau_ i g, h \rangle = \sum_{i \in
   \Z}\langle \tau_i Ag, h \rangle \\ 
&=\sum_{i \in \Z}\langle Ag, \tau_{-i}h \rangle = 
\sum_{i \in \Z}\langle Ag, \tau_i h \rangle
=-\ll Ag, Sh \gg_{-1}.
\end{align*}
\end{proof}

\begin{lemma}\label{lem:anti2ch3}
For all $g \in \Cy$, $\ll Sg, W^A_{0,1}
\gg_{-1}=-\ll  Ag, W^S_{0,1} \gg_{-1}$. 
\end{lemma}
\begin{proof}
By the first identity of (\ref{eq:keyidch3}),
\begin{align*}
\ll Sg, W^A_{0,1} \gg_{-1}& 
=-\ll g, W^A_{0,1} \gg_{*}= -\sum_{i \in \Z}\langle  \tau_i g,
W^A_{0,1} \rangle\\ 
&=-\sum_{i \in \Z}\langle g, W^A_{i,i+1} \rangle = 
-\sum_{i \in \Z}i \langle g, W^A_{i-1,i}-W^A_{i,i+1}  \rangle \\
&=-\sum_{i \in \Z} i \langle g, A \E_i \rangle 
=\sum_{i \in \Z}i \langle Ag, \E_i \rangle
=-\ll Ag, W^S_{0,1} \gg_{-1}.
\end{align*}
\end{proof}

\begin{lemma}\label{lem:anti3ch3}
For all $a \in \R$ and $g \in \Cy$, 
\[
\ll aW^S_{0,1}+S g, aW^A_{0,1}+Ag \gg_{-1}=0.
\]
\end{lemma}
\begin{proof}
By the second identity of (\ref{eq:keyidch3}), it is easy to see that
$\ll W^S_{0,1}, W^A_{0,1} \gg_{-1}=0$. Then, Lemma \ref{lem:antich3}
and Lemma \ref{lem:anti2ch3} conclude the proof straightforwardly.  
\end{proof}
\begin{proposition}\label{prop:aboundch3}
There exists a positive constant $C$ such that for all $g \in \Cy$,
$$
\triple Ag \triple_{-1}^2 \le C \triple S g \triple_{-1}^2 .
$$ 
\end{proposition}
\begin{proof}
By Proposition \ref{prop:directch3}, we have the following variational formula for $\triple Ag \triple_{-1}^2$,
\begin{align*}
\triple Ag \triple_{-1}^2 & = \sup_{a \in \R, h \in \Cy} \frac{ \ll
  Ag, a W^S_{0,1} +Sh  \gg_{-1}^2} { \triple a W^S_{0,1} +S h
  \triple_{-1}^2}   \\ 
& = \max \left\{ \sup_{h \in \Cy}  \frac{ \ll Ag, Sh  \gg_{-1}^2} {
    \triple Sh  \triple_{-1}^2  }  , \sup_{a \neq 0, h \in \Cy} 
  \frac{ \ll Ag, a W^S_{0,1} +Sh  \gg_{-1}^2} { \triple a W^S_{0,1}
    +Sh  \triple_{-1}^2} \right\}  \\ 
& = \max \left\{ \sup_{h \in \Cy}  \frac{ \ll Ag, Sh  \gg_{-1}^2} {
    \triple Sh  \triple_{-1}^2  }  , \ \sup_{h \in \Cy} 
  \frac{ \ll Ag,  W^S_{0,1} +Sh  \gg_{-1}^2} { \triple W^S_{0,1}
    +Sh  \triple_{-1}^2} \right\} . 
\end{align*}
By Lemma \ref{lem:cyboundch3} in Section \ref{sec-sec}, there exists a
positive constant $C$ such that 
$ \ll Ag, Sh \gg_{-1} ^2 \le C \triple  Sh \triple_{-1}^2 \triple Sg
\triple_{-1}^2$ for all $g, h \in \Cy$. Therefore, we have  
\begin{displaymath}
\sup_{h \in \Cy} \frac{ \ll Ag, Sh  \gg_{-1}^2} { \triple Sh  \triple_{-1}^2  }  \le C \triple Sg  \triple_{-1}^2. 
\end{displaymath}
 On the other hand, by Lemma \ref{lem:anti2ch3}, we have 
$$
\ll Ag, W^S_{0,1} \gg_{-1}^2 = \ll Sg, W^A_{0,1} \gg_{-1}^2
\le \triple Sg \triple_{-1}^2 \triple W^A_{0,1}
\triple_{-1}^2.
$$
 Therefore,
 \begin{align*}
 \sup_{ h \in \Cy}  \frac{ \ll Ag, W^S_{0,1} + Sh  \gg_{-1} ^2}{ \triple  W^S_{0,1} +Sh  \triple_{-1}^2  }  
  \le  \triple Sg  \triple_{-1}^2 
\sup_{ h \in \Cy}  \frac{ 2 \triple W^A_{0,1} \triple_{-1}^2  
+2 C \triple Sh  \triple_{-1}^2} { \triple W^S_{0,1} +Sh  \triple_{-1}^2}.
\end{align*}
Now, we only have to show that 
\begin{displaymath}
\sup_{ h \in \Cy} \frac{1} { \triple  W^S_{0,1} +Sh  \triple_{-1}  } 
 < \infty, \qquad   \sup_{ h \in \Cy}  \frac{ \triple Sh
   \triple_{-1}} { \triple  W^S_{0,1} +Sh  \triple_{-1}  } < \infty. 
\end{displaymath}
The first inequality follows from Corollary \ref{cor:directgradient}. To prove the second identity, since we have the first inequality, it is enough to show that
\begin{align*}
\sup_{ \substack{ t \ge 2, h \in \Cy \\   \triple Sh
   \triple_{-1}^2 = t \triple  W^S_{0,1}  \triple_{-1}^2 } } \Big\{
\frac{ \triple Sh \triple_{-1}} { \triple  W^S_{0,1} +Sh  \triple_{-1}  } 
  \Big\} < \infty.
\end{align*}
The triangle inequality shows that 
\[
\triple  W^S_{0,1} +Sh  \triple_{-1} \ge   \triple Sh \triple_{-1}
 -\triple  W^S_{0,1}  \triple_{-1} =(\sqrt{t}-1) \triple  W^S_{0,1}  \triple_{-1}
\]
 for any $h$ satisfying $ \triple Sh\triple_{-1}^2 = t \triple
 W^S_{0,1}  \triple_{-1}^2 $. Then, we obtain that  
\begin{align*}
 \sup_{ \substack{ t \ge 2, h \in \Cy \\  \triple Sh
   \triple_{-1}^2 = t \triple  W^S_{0,1}  \triple_{-1}^2 } } \Big\{ \frac{ \triple Sh \triple_{-1}} { \triple  W^S_{0,1} +Sh  \triple_{-1}  }  \Big\} \le \sup_{t \ge 2} \{ \frac{t}{(\sqrt{t}-1)^2} \} < \infty.
\end{align*}
\end{proof}

Now, we have all elements to show the desired decomposition of the
Hilbert spaces $\mathcal{H}_{-1}$. 

\begin{proposition}\label{prop:generatech3}
Denote by $L \Cy$ the space $\{ Lg; \ g \in \Cy \}$. Then, we have
\begin{displaymath}
\mathcal{H}_{-1} = \overline{L\Cy}|_{\mathcal{N}} + \{  p_1^2 - p_0^2\}.
\end{displaymath}
\end{proposition}
\begin{proof}
Since $\{  p_1^2 - p_0^2 \}$ and $L\Cy$ are contained in
$\Cy_0$ by definition, $\mathcal{H}_{-1}$ contains the right
hand space. To prove the converse inclusion, 
let $h \in \mathcal{H}_{-1}$ so that 
$\ll h,  p_1^2 -p_0^2 \gg_{-1}=0$ and 
$\ll h, L g \gg_{-1}=0$ for all $g \in \Cy$. By Corollary
\ref{cor:directgradient}, $h=\lim_{k \to
\infty}S^{\gamma}g_k$ in $\mathcal{H}_{-1}$ for some sequence 
$g_k \in \Cy$. 
Namely, 
$$
\triple h \triple_{-1}^2 = \lim_{k \to \infty} \ll Sg_k,
S g_k \gg_{-1} = \lim_{k \to \infty} \ll S g_k, L g_k \gg_{-1}
$$ since 
$\ll S g_k, Ag_k \gg_{-1}=0$ by Lemma \ref{lem:antich3}. 
On the other hand, by the assumption $\ll h, Lg_k \gg_{-1}=0$ for all
$k$. Also, by Proposition \ref{prop:aboundch3}, 
$$
\sup_k \triple L g_k \triple_{-1} \le (C+1) \sup_k \triple
S g_k \triple_{-1} :=C_h
$$ 
is finite. Therefore, 
\begin{equation*}
  \begin{split}
    \triple h\triple_{-1}^2 = \lim_{k \to \infty} \ll S g_k, L g_k
    \gg_{-1} =\lim_{k \to \infty} \ll Sg_k-h, Lg_k \gg_{-1} \\
    \le \lim_{k \to\infty} C_h \triple Sg_k - h \triple_{-1} =0 .
  \end{split}
\end{equation*}
This concludes the proof. 
\end{proof}

\begin{lemma}\label{lem:decompositionch3}
We have
\begin{displaymath}
\mathcal{H}_{-1} = \overline{L\Cy}|_{\mathcal{N}} \oplus \{  p_1^2 - p_0^2 \}.
\end{displaymath}
\end{lemma}
\begin{proof}
Let a sequence $g_k \in \Cy$ satisfy 
$\lim_{k \to\infty} L g_k=a (p_1^2 - p_0^2)$ in
$\mathcal{H}_{-1}$ for some $a \in \R$. By a similar
argument of the proof of Proposition \ref{prop:generatech3},
\begin{equation*}
  \begin{split}
    \limsup_{k \to \infty} \ll S g_k, S g_k \gg_{-1}=\limsup_{k \to
      \infty} \ll Lg_k, Sg_k \gg_{-1} \\
    = \limsup_{k \to \infty}\ll Lg_k-a(p_1^2 - p_0^2), Sg_k \gg_{-1}=0
  \end{split}
\end{equation*}
since $\ll p_1^2 - p_0^2, Sg_k \gg_{-1}=0$ for all
$k$. On the other hand, by Proposition \ref{prop:aboundch3}, 
$\triple Lg_k\triple_{-1}^2 \le (C +1) \triple Sg_k \triple_{-1}^2$, 
then $a=0$. 
\end{proof}

\begin{corollary}\label{cor:decompositionch3}
For each $g \in \Cy_0$, there exists a unique constant $a \in \R$ such
that 
\begin{displaymath}
g+ a( p_1^2 - p_0^2 ) \in \overline{L\Cy}. 
\end{displaymath}
\end{corollary}

By this corollary,  it is obvious that
there exists a sequence of local functions $F_K$ in $\Cy$ such that 
 $$
\triple W_{0,1} + \kappa_\beta (p_1^2 - p_0^2) - L^\gamma F_K \triple_{-1}
\to 0
$$ 
as $K\to \infty$ which conclude the first statement of Lemma \ref{lem:CLT}. We prove the rest of the claim of Lemma \ref{lem:CLT} in Proposition \ref{prop:CLT2} in Section \ref{sec-diff-coeff}.

\section{Boltzmann-Gibbs principle}\label{sec-bgp}

In this section, we prove Lemma \ref{lem:Boltzmann-Gibbs}. 
First, recall that for each $E >0$, $\beta(E)$ is the inverse of the
function $E(\beta) = \bar\E_\beta = \frac 1{2\beta} + \tilde{V}(\beta)$,
where
\newline
$\tilde{V}(\beta) = <V(r_0)>_\beta$.
Then, by simple calculations, we have
\begin{equation*}
  \begin{split}
    \frac{d}{d E} \langle p_1^2 \rangle_{\beta(E)} &=\frac{d}{d E} \Big
    (\frac{1}{\beta(E)} \Big) =\frac{-1}{\beta^2}
    \left(\frac{dE(\beta)}{d\beta}\right)^{-1} \\&=\frac{-1}{\beta(E)^2}
    \left\{ -\frac{1}{2\beta(E)^2}+\tilde{V}'(\beta(E)) \right\} ^{-1}
    = \frac{1}{\chi(\beta(E))\beta(E)^2}.
  \end{split}
\end{equation*}
Now, we can rewrite the term $I^3_{N,t} (H)$ as 
\begin{displaymath}
\int^t_0 \frac{1}{\sqrt{N}}\sum_{i \in \T_N} H''_N (\frac{i}{N})[ p_i^2 - h(E) -h'(E)( \E_i - E ) ]\; ds
\end{displaymath}
where $h(E)= \frac{1}{\beta(E)}=\langle p_1^2 \rangle_{\beta(E)}$, and
$H''_N = \Delta^N H$ . 
Lemma \ref{lem:Boltzmann-Gibbs} follows from standard arguments
(cf. \cite{KL}). We sketch it here for completeness.

Here one can introduce a further average in a microscopic block of
length $K<<N$ and substitute the following expression
\begin{equation*}
  \int^t_0 \frac{1}{\sqrt{N}}\sum_{i \in \T_N} H''_N (\frac{i}{N}) 
  \tau_i \varphi_K(\omega_s) \; ds
\end{equation*}
with $\varphi_K = \frac 1{2K+1}\sum_{|j|\le K} [ p_j^2 - h(E)
-h'(E)( \E_j - E ) ]$.  
Using Schwarz inequality, the difference can be estimated to be of
order $K/N$. 

Define $\hat\varphi_K = \varphi_K - \langle \varphi_K
\rangle_{\Lambda_K,\bar{\E}_K}$, with $\bar{\E}_K = (2K+1)^{-1}
\sum_{j\in\Lambda_K} \E_j$.
By \eqref{eq-genbound}
\begin{equation}\label{bbb}
  \begin{split}
    \mathbb E_{\nu_\beta} &\left(\sup_{0\le t \le T} \left[ \int^t_0
        \frac{1}{\sqrt{N}}\sum_{i \in \T_N} H''_N (\frac{i}{N})
        \tau_i \hat\varphi_K (\omega_s) \; ds\right]^2 \right)\\
   & \le 16 T \sup_f \left\{  \frac{1}{\sqrt{N}}\sum_{i \in
        \T_N} H''_N (\frac{i}{N})  \langle f\tau_i
      \hat\varphi_K\rangle - N^2 \gamma <f, (-S_N) f> \right\}.
  \end{split}
\end{equation}

By the spectral gap on $-S_{\Lambda_K}$ (see section
\ref{sec:spectral-gap-bound}),  we can define $U_K =
(-S_{\Lambda_K})^{-1} \hat\varphi_K$.
Then
\begin{equation*}
  \begin{split}
    \frac{1}{\sqrt{N}}\sum_{i \in \T_N} H''_N (\frac{i}{N}) & \langle
    f\tau_i \hat\varphi_K\rangle \\
    &\le \left( \frac{1}{N}\sum_{i \in
        \T_N} H''_N (\frac{i}{N})^2 \langle \hat\varphi_K
      U_K\rangle\right)^{1/2} \left(N <f, (-S_{\Lambda_K})
      f>\right)^{1/2}.
  \end{split}
\end{equation*}
Consequently the right hand side of \eqref{bbb} is bounded by 
\begin{equation*}
  \begin{split}
    16 T \sup_f \Big\{ \left( \frac{1}{N}\sum_{i \in
        \T_N} H''_N (\frac{i}{N})^2  \langle \hat\varphi_K
      U_K\rangle\right)^{1/2} \left(N  <f, (-S_{\Lambda_K})
      f>\right)^{1/2}\\ -  N^2 \gamma <f, (-S_{\Lambda_K}) f>
    \Big\}
    \le \frac{C_K}{N}    
  \end{split}
\end{equation*}
and it follows that
\begin{equation*}
  \lim_{N\to\infty} \mathbb E_{\nu_\beta} \left(\sup_{0\le t \le T}
   \left[ \int^t_0 \frac{1}{\sqrt{N}}\sum_{i \in \T_N} H''_N
     (\frac{i}{N}) \tau_i \hat\varphi_K (\omega_s) \; ds\right]^2
 \right) = 0.
\end{equation*}
We are left to estimate the corresponding term with $\langle \varphi_K
\rangle_{\Lambda_K,\bar{\E}_K}$. Denote $\bar\varphi_K = \langle \varphi_K
\rangle_{\Lambda_K,\bar{\E}_K}$ and observe that it has support on
$\Lambda_k$ and its variance with respect to $\nu_\beta$ 
is of order $K^{-2}$. Then by Schwarz inequality and stationarity of
$\nu_\beta$ 
\begin{equation*}
  \begin{split}
     \mathbb E_{\nu_\beta} \left(\sup_{0\le t \le T}
      \left[ \int^t_0 \frac{1}{\sqrt{N}}\sum_{i \in \T_N} H'_N
        (\frac{i}{N}) \tau_i \bar\varphi_K(\omega_s)
        \; ds\right]^2 \right) \\
    \le CT^2 \langle  \left[ \frac{1}{\sqrt{N}}\sum_{i \in \T_N} H''_N
        (\frac{i}{N}) \tau_i \bar\varphi_K\right]^2 \rangle \\
        = \frac{CT^2}{N} \sum_{i,j} H''_N (\frac{i}{N}) H''_N
        (\frac{j}{N}) 
        \langle  \tau_{i-j} \bar\varphi_K \bar\varphi_K \rangle\\
        \le \frac{CT^2}{2N} \sum_{i,j} (H''_N (\frac{i}{N})^2 +
        H''_N(\frac{j}{N})^2) 
        \langle  \tau_{i-j} \bar\varphi_K \bar\varphi_K \rangle\\
        = \frac{CT^2}{N} \sum_{i} H''_N (\frac{i}{N})^2 
        \sum_j \langle  \tau_{j} \bar\varphi_K \bar\varphi_K \rangle\\
        \le \frac{C'T^2 }{N} \sum_{i} H''_N (\frac{i}{N})^2 
        K \langle \bar\varphi_K^2\rangle
  \end{split}
\end{equation*}
that goes to $0$ as $K\to\infty$, uniformly in $N$.

\section{Tightness}
\label{sec:tightness}

The argument exposed above proves the convergence of the finite
dimensional distribution of $\mathbb Q_N$. In order to conclude the
proof of Theorem \ref{thm:main}, we need to prove the tightness of the
sequence in $C([0,T], \mathfrak H_{-k})$. The argument we use is standard.
 We report it here for completeness.

Compactness follows from the following two statements:
\begin{equation}
  \label{eq:19}
  \lim_{A\to\infty} \limsup_{N\to\infty} \mathbb P_{\nu_\beta^N}\left(
    \sup_{0\le t \le T}\|Y^N_t\|_{-k} \ge A \right) = 0,
\end{equation}
\begin{equation}
  \label{eq:21}
   \lim_{\delta\to 0} \limsup_{N\to\infty} \mathbb P_{\nu_\beta^N}\left(
     w_{-k}(Y^N,\delta) \ge \epsilon \right) = 0
  \qquad \forall \epsilon>0 .
\end{equation}
where $\|\cdot\|_{-k}$ is the norm in $\mathfrak H_{-k}$ and
$w_{-k}(Y^N,\delta)$ is the corresponding modulus of continuity in 
$C([0,T], \mathfrak H_{-k})$. We recall that  $\|\cdot\|_{-k}$ can be
written as
\begin{equation*}
   \| Y\|_{-k}^2 = \sum_{n\ge 1} (\pi n)^{-2k} |Y(e_n)|^2 
\end{equation*}
with $e_n(x) = \sqrt 2 \sin(\pi n x)$.

Recall the decomposition of $Y_t^N$ given by (\ref{eq:decomp}):
\begin{equation*}
Y^N_t(e_n) =  Y^N_0(e_n)  +   (\pi n)^2 D \int^t_0 Y^N_s ( e_n) ds +
M^1_{N,F_K,t}(e_n) + Z_{N,F_K,t}(e_n)
\end{equation*}
where $E_{\nu^N_\beta}\left(\sup_{0\le t\le T} (Z_{N,F_K,t}(e_n))^2 \right)$ can be estimated by the proof of Lemmas \ref{lem:rest}, \ref{lem:Boltzmann-Gibbs}. On the other hand, $E_{\nu^N_\beta}\left( (M^1_{N,F_K,t}(e_n))^2 \right)$ can be computed explicitly.
Then, for $k > 5/2$, (\ref{eq:19}) and (\ref{eq:21}) follows by
standard arguments (cf. \cite{KL}). 

\section{Sector property of $A$ in $\mathcal{H}_{-1}$}
\label{sec-sec}

In this section, we show that the operator $A$ satisfy a sector
condition in $\mathcal{H}_{-1}$.
First, we prepare a useful lemma.

\begin{lemma}\label{lem:gamma1ch3}
For all $f,g\in \Cy$, 
$$
\ll  Sf, Ag \gg_{-1} = - \langle \Ga_f (X_0 -Y_{0,1}) \Ga_g
\rangle = \langle \Ga_g (X_0 -Y_{0,1}) \Ga_f \rangle .
$$
\end{lemma}

\begin{proof}
By the first identity of (\ref{eq:keyidch3}),
\begin{align*}
- & \ll Sf, Ag \gg_{-1} = \ll f, Ag \gg_{*} =
\sum_{i \in \Z}\langle  \tau_i f, Ag \rangle = \sum_{i,j \in
  \Z}\langle  \tau_i f, (X_j -Y_{j,j+1}) g \rangle\\ 
&=\sum_{i,k \in \Z}\langle  \tau_i f, (X_{k+i} -Y_{k+i,k+i+1})  g \rangle
=\sum_{i,k \in \Z}\langle  \tau_i f, \tau_i ( (X_k-Y_{k,k+1})
\tau_{-i} g) \rangle \\
&= \sum_{k \in \Z}\langle  f,(X_k-Y_{k,k+1}) \Ga_g \rangle = \sum_{k
  \in \Z} \langle \tau_k f, (X_0-Y_{0,1})  \Ga_g \rangle =\langle
\Ga_f (X_0-Y_{0,1}) \Ga_g \rangle.  
\end{align*}
\end{proof}

Define $\tilde\Gamma_f$ as $\tilde\Gamma_f= \sum_{|i|\le s_f+1} \tau_i f$.
Observe that in the above expression, one can restrict the definition of
$\Gamma_f$ and $\Gamma_g$ as finite sums, namely, we can replace them by $\tilde\Gamma_f$ and $\tilde\Gamma_g$.
Decompose now $\tilde\Gamma_f = \tilde\Gamma_f^e + \tilde\Gamma_f^o$, where $\tilde\Gamma_f^e$
is even in $p_0$ and $\tilde\Gamma_f^o$ is odd. Observe that the vector fields 
 $X_0$ and $Y_{0,1}$ change the parity of $p_0$, so we have
 \begin{equation*}
     \langle \tilde\Ga_f (X_0-Y_{0,1}) \tilde\Ga_g \rangle = 
     \langle \tilde\Ga_f^o (X_0-Y_{0,1}) \tilde\Ga_g^e \rangle - \langle
\tilde\Ga_g^o (X_0-Y_{0,1}) \tilde\Ga_f^e \rangle   .  
 \end{equation*}

Applying Schwarz inequality, we can bound the above expression by 
\begin{equation*}
  \label{eq:11}
  \langle (\tilde\Ga_f^o)^2\rangle^{1/2} \langle[(X_0-Y_{0,1})
  \tilde\Ga_g^e]^2 
  \rangle^{1/2}  + \langle (\tilde\Ga_g^o)^2\rangle^{1/2}
  \langle[(X_0-Y_{0,1}) \tilde\Ga_f^e]^2 \rangle^{1/2}
\end{equation*}
and applying the spectral gap for $X_0^2$ plus Schwarz inequality again, the last
term is bounded by 
\begin{equation*}
 C\{ \langle (X_0\tilde\Ga_f^o)^2\rangle^{1/2} \langle(X_0 \tilde\Ga_g^e)^2 +
  (Y_{0,1} \tilde\Ga_g^e)^2 
  \rangle^{1/2}  + \langle (X_0\tilde\Ga_g^o)^2\rangle^{1/2}
  \langle(X_0 \tilde\Ga_f^e)^2 +
  (Y_{0,1} \tilde\Ga_f^e)^2  \rangle^{1/2} \}
\end{equation*}
with some positive constant $C$.

Recall that 
\begin{displaymath}
\triple Sf \triple_{-1}^2 = \frac{\gamma}{2}
\langle (Y_{0,1} \tilde\Gamma_f)^2 + ( X_{0} \tilde\Gamma_f )^2 \rangle = \frac{\gamma}{2} \langle (Y_{0,1} \tilde\Gamma_f^e)^2 + ( Y_{0,1} \tilde\Gamma_f^o )^2 + (X_{0} \tilde\Gamma_f^e)^2 + ( X_{0} \tilde\Gamma_f^o )^2 \rangle.     
\end{displaymath}
Then we obtain the sector condition:

\begin{lemma}[sector condition]\label{lem:cyboundch3}
There exists a positive constant $C$ such that for all $f,g\in \Cy$, 
\begin{equation*}
| \ll Sf, Ag \gg_{-1} | \le C \triple Sf \triple_{-1} \triple Sg \triple_{-1}. \label{eq:sector3ch3}
\end{equation*}
\end{lemma}

\section{Closed forms}
\label{sec:closedformch2}

In this section, to complete the proof of Lemma \ref{lem:clt2}, we
introduce the notion of closed form and give a characterization for it.  
We generalize some ideas developed in the PhD thesis of Hernandez
\cite{H}, where microcanonical surfaces were given by spheres,
following the general setup of the seminal work of Varadhan \cite{V}
(see also \cite{KL}, appendix 3, section 4). The nonlinearity of our
interaction reflected in the non-constant curvature of our
microcanonical manifolds, requires some substantial modification of
the original approach.   We will make use of the spectral gap estimate
proved in section \ref{sec:spectral-gap-bound}. 


Let us decompose $\Cy =\cup_{k\geq 1} \Cy_k$, where
$\Cy_k$ is the space of functions $F\in \Cy$ depending only on
the variables $(p_i,r_i)_{-k \le i \le k}$. Given $F\in\Cy_k$
recall the definition of the formal sum
\[
\Gamma_{F}(p,r)=\sum_{j=-\infty}^{\infty}\tau_jF(p,r)
\]
and that for every $i\in\Z$ the expressions
\[
\frac{\partial\Gamma_{F}}{\partial p_i}(p,r)=\sum_{i-k\leq j\leq i+k}\frac{\partial}{\partial p_i}\tau_jF(p,r)
\]
and
\[
\frac{\partial\Gamma_{F}}{\partial r_i}(p,r)=\sum_{i-k\leq j\leq i+k}\frac{\partial}{\partial r_i}\tau_jF(p,r)
\]
are well defined.
The formal invariance $\Gamma_{F}(\tau_i(p,r))=\Gamma_{F}(p,r)$ leads us to the relation
\begin{equation}\label{covariance}
\frac{\partial\Gamma_{F}}{\partial p_i}(p,r)=\frac{\partial\Gamma_{F}}{\partial p_0}(\tau_i (p,r)).
\end{equation}

Remember that $Y_{i,j}=p_i\partial_{r_j}-V^{\prime}(r_j)\partial_{p_i}$ and $X_{i}:=Y_{i,i}$. Given $F\in\Cy$ and $i \in\mathbb{Z}$, $X_{i}(\Gamma_{F})$ and $Y_{i,i+1}(\Gamma_{F})$ are well defined and satisfy
\[
X_{i}(\Gamma_{F})(p,r)=\tau_{i}X_{0}(\Gamma_{F})(p,r), \quad Y_{i,i+1}(\Gamma_{F})(p,r)=\tau_{i}Y_{0,1}(\Gamma_{F})(p,r).
\]
Now we consider the following linear space
\[
\mathcal{B} = \{(X_{0}(\Gamma_{F}),Y_{0,1}(\Gamma_{F}))\in L^2(\nu_{\beta}) \times L^2(\nu_{\beta}): F\in \Cy \}.
\]

We denote by $\mathfrak{H}$ the linear space generated by the closure
of $\mathcal{B}$ in $L^2(\nu_{\beta}) \times L^2(\nu_{\beta})$
and $(0,p_0V^{\prime}(r_1))$
\begin{equation}
  \label{eq:12}
  \mathfrak{H} \ = \ \overline{\mathcal{B} + \{(0, p_0 V^{\prime}(r_1))\}}.
\end{equation}

First, we observe that defining a vector-valued function
$\xi=(\xi^0,\xi^1)$ as $(X_{0}(\Gamma_{F}),Y_{0,1}(\Gamma_{F}))$ for
$F\in \Cy$ or $(0,p_0V^{\prime}(r_1))$, the following properties are
satisfied: 
\begin{enumerate}[\it i)]
\item $X_{i}(\tau_j\xi^0)=X_{j}(\tau_i\xi^0)$ \quad \text{for all}
  \quad  $ i,j \in \Z$, \label{cond-closed}
\item $Y_{i,i+1}(\tau_j\xi^1)=Y_{j,j+1}(\tau_i\xi^1)$ \quad \text{for all} \quad  $i,j \in \Z$,
\item $X_{i}(\tau_j\xi^1)=Y_{j,j+1}(\tau_i\xi^0)$ \quad \text{if} \quad  $\{i\}\cap\{j,j+1\}=\emptyset$, 
\item $p_{i}[X_{i}(\tau_i\xi^1)-Y_{i,i+1}(\tau_{i}\xi^0)]=V^{\prime}(r_{i+1})\tau_i\xi^0-V^{\prime}(r_i)\tau_{i}\xi^1$\ \text{for all} $i\in \mathbb{Z}$,
\item $V^{\prime}(r_{i+1})[X_{i+1}(\tau_i\xi^1)-Y_{i,i+1}(\tau_{i+1}\xi^0)]=V^{\prime \prime}(r_{i+1})p_{i+1}\tau_i\xi^1-V^{\prime \prime}(r_{i+1})p_{i}\tau_{i+1}\xi^0$\ \ \text{for all} $i\in \mathbb{Z}$.
\end{enumerate}

We call a weakly closed form (or \emph{germ} of a weakly closed form,
cf. \cite{KL}), a couple of functions 
$\xi =(\xi^0,\xi^1)\in L^2(\nu_{\beta})\times L^2(\nu_{\beta})$, that
satisfy   \ref{cond-closed}) to v) in a weak sense. A smooth
approximation of weakly closed form is not necessarily closed, and some
type of Hodge decomposition is needed. This will be done only after
localization in the proof of the following theorem that is the main
result of this section.  

\begin{theorem}\label{thm:closedformch3}
If $\xi=(\xi^0,\xi^1)\in L^2(\nu_{\beta})\times L^2(\nu_{\beta})$
satisfies conditions \ref{cond-closed}) to v) in a weak sense, then $\xi\in\mathfrak{H}$.
\end{theorem}
\begin{proof}
The goal is to find a sequence $(F_L)_{L\geq 1}$ in $\Cy$ such that 
$$
(\xi^0-X_0(\Gamma_{F_L}),\xi^1-Y_{0,1}(\Gamma_{F_L})) 
\ \mathop{\longrightarrow}_{L\to\infty}\ (0,cp_0V^{\prime}(r_1))
$$
in $L^2(\nu_{\beta}) \times L^2(\nu_{\beta})$
for some constant $c$.

First, observe that for a function $F \in \Cy_k$ we can
rewrite, by using 
(\ref{covariance}), 
\begin{equation}\label{suminvgrad} 
X_{0}(\Gamma_{F}) = \sum_{i=-k}^{k}X_{i}(F)(\tau_{-i}(p,r))
\end{equation}
and 
\begin{equation}\label{suminvgrad2}
  \begin{split}
    Y_{0,1}(\Gamma_{F}) =
    \sum_{i=-k}^{k-1}Y_{i,i+1}(F)(\tau_{-i}(p,r))-
    \left(V^{\prime}(r_{k+1})\frac{\partial F}{\partial p_{k}} \right)
    (\tau_{-k}p)\\
    +\left(p_{-k-1}\frac{\partial F}{\partial r_{-k}}
    \right)(\tau_{k+1}(p,r)) .
  \end{split}
\end{equation}

We define for $m=0,1$
\[
\xi_{i}^{m,(L)}=\textbf{E}_{\nu_{\beta}}[\xi_{i}^m|\mathcal{F}_{L}]\varphi\left(\frac{1}{2L+1}\sum_{i=-L}^{L} \{\frac{p_i^2}{2}+V(r_i)\}
\right)
\]
where $\xi_{i}^m(p,r)=\tau_i\xi^m(p,r)$, $\mathcal{F}_{L}$ is the sub
$\sigma$-field of $\Omega$ generated by $(p_i,r_i)_{i=-L}^L$ and
$\varphi$ is a smooth positive function with compact support such that
$\varphi(E(\beta))=1$ and bounded by $1$ (we need this cutoff in order to do uniform
bounds later). 
Because $\nu_{\beta}$ is a product measure and $\varphi$ satisfies
that   
\[
X_{i}\varphi\left(\frac{1}{2L+1}\sum_{i=-L}^{L} \{\frac{p_i^2}{2}+V(r_i)\} \right)=0
\]
for $-L \le i \le L$ and 
\[
Y_{i,i+1}\varphi\left(\frac{1}{2L+1}\sum_{i=-L}^{L} \{\frac{p_i^2}{2}+V(r_i)\} \right)=0
\]
for $-L \le i \le L-1$, the set of functions
$\{(\xi_{i}^{0,L})\}_{-L\leq i\leq  L}$ and
$\{(\xi_{i}^{1,L})\}_{-L\leq i\leq  L-1}$ even satisfies the
conditions {\it\ref{cond-closed})} to {\it v)} on the finite set
$\{-L,-L+1, \dots, L \}$ if we replace $\tau_i \xi^0$ by $\xi_i^{0,(L)}$ and $\tau_i \xi^1$ by $\xi_i^{1,(L)}$. Therefore,
they define a closed form in a weak sense on a finite dimensional space. 
To obtain a closed form on each microcanonical
manifold $\left\{\omega \in (\R^2)^{2L+1}; \sum_{-L}^L \mathcal E_i = (2L+1) E
\right\}$, we first take a smooth $\mathcal{F}_L$-measurable approximation of $\{(\xi_{i}^{0,L})\}_{-L\leq i\leq  L}$ and
$\{(\xi_{i}^{1,L})\}_{-L\leq i\leq  L-1}$ in $L^2(\nu_{\beta})$, and denote it by $\{(\zeta_i^{0,(L)} ) \}_{-L \le i \le L}$ and  $\{(\zeta_i^{1,(L)} ) \}_{-L \le i \le L-1}$. For arbitrary chosen $\epsilon_L >0$, we choose $\{(\zeta_{i}^{0,L})\}_{-L\leq i\leq  L}$ and
$\{(\zeta_{i}^{1,L})\}_{-L\leq i\leq  L-1}$ satisfying
\begin{align*}
\sum_{i=-L}^{L} \|\xi_i^{0,(L)} - \zeta_i^{0,(L)} \|_{L^2(\nu_{\beta})}^2 +  \sum_{i=-L}^{L-1} \| \xi_i^{1,(L)} - \zeta_i^{1,(L)}  \|_{L^2(\nu_{\beta})}^2 \le \frac{\epsilon_L}{4}.
\end{align*}
Unfortunately, this smooth approximation may not be closed. Therefore, to obtain a smooth closed form, we consider a one-form $\zeta=\sum_{i=-L}^L \zeta_i^{0,(L)}dX_i + \sum_{i=-L}^{L-1} \zeta_i^{1,(L)}dY_{i,i+1}$ on each microcanonical manifold. As shown in Section \ref{sec:lie-algebra},
 $\text{Lie}\{ \{X_i, i=-L, \dots, L \} \{ Y_{i,i+1}, i=-L, \dots, L-1 \} \}$ generates
the all tangent space of each microcanonical
manifold, so $\zeta$ is well-defined on any chart. By the Hodge decomposition (cf. \cite{M}) with respect to a Riemannian
structure associated to our microcanonical measure, there exists a smooth function $g$ and a smooth two-form $H$ satisfying
\begin{align*}
\zeta=dg+\delta H
\end{align*}
since $0$ is the only harmonic function on each of these manifolds by the assumptions on $V$. Let $H_i^{0}$ and $H_i^1$ given by $\delta H=\sum_{i=-L}^L H_i^{0}dX_i + \sum_{i=-L}^{L-1} H_i^{1}dY_{i,i+1}$. Since $dg$ and $\delta H$ are orthogonal and the set of functions$\{(\xi_{i}^{0,L})\}_{-L\leq i\leq  L}$ and
$\{(\xi_{i}^{1,L})\}_{-L\leq i\leq  L-1}$ is closed,
\begin{align*}
& \textbf{E}_{\nu_{\beta}}\left[\sum_{i=-L}^L (H_i^{0})^2 +
  \sum_{i=-L}^{L-1} (H_i^{1})^2\right]
=\textbf{E}_{\nu_{\beta}}\left[\sum_{i=-L}^L H_i^{0}\zeta_i^{0,(L)} +
\sum_{i=-L}^{L-1} H_i^{1}\zeta_i^{1,(L)}\right] \\ 
&= \textbf{E}_{\nu_{\beta}}\left[\sum_{i=-L}^L H_i^{0} ( \zeta_i^{0,(L)} -\xi_i^{0,(L)} )  + \sum_{i=-L}^{L-1} H_i^{1} ( \zeta_i^{1,(L)}-\xi_i^{1,(L)} ) \right] \\
& \le \textbf{E}_{\nu_{\beta}}\left[\sum_{i=-L}^L (H_i^{0})^2 +
    \sum_{i=-L}^{L-1} (H_i^{1})^2\right]^{1/2}  \ 
\textbf{E}_{\nu_{\beta}}\left[\sum_{i=-L}^L ( \zeta_i^{0,(L)} -\xi_i^{0,(L)} )^2 + \sum_{i=-L}^{L-1} (\zeta_i^{1,(L)}-\xi_i^{1,(L)})^2\right]^{1/2}.
\end{align*}
Therefore, we know that 
\begin{align*}
\textbf{E}_{\nu_{\beta}}[\sum_{i=-L}^L (H_i^{0})^2 &+ \sum_{i=-L}^{L-1} (H_i^{1})^2]  = \sum_{i=-L}^L \|X_i g- \zeta_i^{0,(L)} \|_{L^2(\nu_{\beta})}^2 + \sum_{i=-L}^{L-1} \|Y_{i,i+1} g - \zeta_i^{1,(L)}  \|_{L^2(\nu_{\beta})} ^2\\
&  \le \sum_{i=-L}^{L} \|\xi_i^{0,(L)} - \zeta_i^{0,(L)} \|_{L^2(\nu_{\beta})}^2 +  \sum_{i=-L}^{L-1} \| \xi_i^{1,(L)} - \zeta_i^{1,(L)}  \|_{L^2(\nu_{\beta})}^2 \le \frac{\epsilon_L}{4}.
\end{align*}

Therefore,  
\begin{align*}
& \sum_{i=-L}^{L} \|X_i g- \xi_i^{0,(L)} \|_{L^2(\nu_{\beta})}^2 +  \sum_{i=-L}^{L-1} \|Y_{i,i+1} g - \xi_i^{1,(L)}  \|_{L^2(\nu_{\beta})}^2  \\
&  \le 2 \Big\{ \sum_{i=-L}^{L} \|X_i g- \zeta_i^{0,(L)} \|_{L^2(\nu_{\beta})}^2 +  \sum_{i=-L}^{L-1} \|Y_{i,i+1} g - \zeta_i^{1,(L)}  \|_{L^2(\nu_{\beta})}^2  \Big\} \\
& + 2 \Big\{ \sum_{i=-L}^{L} \|\zeta_i^{0,(L)} - \xi_i^{0,(L)} \|_{L^2(\nu_{\beta})}^2 +  \sum_{i=-L}^{L-1} \|\zeta_i^{1,(L)} - \xi_i^{1,(L)}  \|_{L^2(\nu_{\beta})}^2  \Big\} \le \epsilon_L.
\end{align*}
From now on, we show that we can choose a $\mathcal{F}_L$-measurable function $g^{(L)}$ which is smooth on 
each microcanonical manifold
(with respect to the vector fields of the tangent space) and
satisfies
\begin{equation}\label{system}
\begin{split}
X_{i}(g^{(L)})&=\xi_{i}^{0,(L)}+ \epsilon_{i}^{0,(L)} \ \quad  \textit{for} \quad -L \leq i \leq L, \\
Y_{i,i+1}(g^{(L)})&=\xi_{i}^{1,(L)} + \epsilon_{i}^{1,(L)} \ \quad  \textit{for} \quad -L \leq i \leq L-1
\end{split}
\end{equation}
where $\sum_{i=-L}^L \textbf{E}_{\nu_{\beta}}[(\epsilon_{i}^{0,(L)})^2]+\sum_{i=-L}^{L-1} \textbf{E}_{\nu_{\beta}}[(\epsilon_{i}^{1,(L)})^2] \le \epsilon_{L}$ for arbitrary given $\epsilon_L >0 $. 
Then, by the spectral gap proved in Section \ref{sec:spectral-gap-bound}, $g^{(L)}$ is in $L^2(\nu_{\beta})$, so we have a sequence of functions $\{g_n\}_{n \in \N}$ in $\Cy_L$ such that 
\begin{equation*}
g_n \to g^{(L)} \  \textit{in} \ L^2 (\nu_{\beta}), \quad  X_i (g_n) \to X_i (g^{(L)}) \ \textit{in} \ L^2 (\nu_{\beta}) \quad \textit{for} \quad -L \leq i \leq L
\end{equation*}
and 
\begin{equation*}
Y_{i,i+1} (g_n) \to Y_{i,i+1} (g^{(L)}) \ \textit{in} \ L^2 (\nu_{\beta}) \quad \textit{for} \quad -L \leq i \leq L-1.
\end{equation*}
It means that we can choose a function $g^{(L)}$ in $\Cy_L$ satisfying
(\ref{system}) at the beginning for arbitrary given $\epsilon_L$.  
From now on, we fix a sequence $\{ \epsilon_{L} \}_{L}$ that $\epsilon_{L}\to 0$ as $L \to \infty$. Observe that
$g^{(L)}-\textbf{E}_{\nu_{\beta}}[g^{(L)}|\mathcal{E}_{-L}+\cdots+\mathcal{E}_L]$
still satisfies (\ref{system}). So we can suppose
that
$\textbf{E}_{\nu_{\beta}}[g^{(L)}|\mathcal{E}_{-L}+\cdots+\mathcal{E}_L=
(2L +1) E]=0$ for every $E>0$. \\
Define
\[
{g}^{(L,k)}=\frac{\beta }{2(L+k)  \phi_{\beta}}\textbf{E}_{\nu_{\beta}}[p^2_{-L-k-1}V^{\prime}(r_{L+k+1})^2g^{(2L)}|\mathcal{F}_{L+k}]
\]
and
\[
\widehat{g}^{L}=\frac{4}{L}\sum_{k=L/2}^{3L/4}{g}^{(L,k)}
\]
where $\phi_{\beta}:=\textbf{E}_{\nu_{\beta}}[V^{\prime}(r_0)^2]$.

Using (\ref{suminvgrad}) and (\ref{suminvgrad2}) for ${g}^{(L,k)}$ and then averaging over $k$ we obtain that
\[
X_{0}\left(\sum_{j=-\infty}^{\infty}\tau_j\widehat{g}^{L}\right)=\xi^0+\frac{\beta}{\phi_{\beta}}[I^1_L+I^2_L+I^3_L+I^4_L+I^5_L]
\]
and
\[
Y_{0,1}\left(\sum_{j=-\infty}^{\infty}\tau_j\widehat{g}^{L}\right)=\xi^1+\frac{\beta}{\phi_{\beta}}[J^1_L+J^2_L+J^3_L+J^4_L-R^1_L+R^2_L],
\]
where
\begin{align*}
I^1_L&=\widehat{\sum_{k=L/2}^{3L/4}}\ \ \widehat{\sum_{i=-L-k}^{L+k-1}}\tau_{-i}\textbf{E}_{\nu_{\beta}}[V^{\prime}(r_{L+k+1})^2p^2_{-L-k-1}(\xi^{0,(2L)}_{i}-\xi^{0,(L+k)}_{i})\varphi(\mathcal{E}_{-2L,2L})|\mathcal{F}_{L+k}], 
\\
I^2_L&=\widehat{\sum_{k=L/2}^{3L/4}}\ \ \widehat{\sum_{i=-L-k}^{L+k-1}}\tau_{-i}\{(\xi^{0,(L+k)}_{i}-\xi_{i}^0)\textbf{E}_{\nu_{\beta}}[V^{\prime}(r_{L+k+1})^2p^2_{-L-k-1}\varphi(\mathcal{E}_{-2L,2L})|\mathcal{F}_{L+k}]\}, 
\\
I^3_L&=\widehat{\sum_{k=L/2}^{3L/4}}\ \ \widehat{\sum_{i=-L-k}^{L+k-1}}\xi^0(p,r)\tau_{-i}\textbf{E}_{\nu_{\beta}}[V^{\prime}(r_{L+k+1})^2p^2_{-L-k-1}(\varphi(\mathcal{E}_{-2L,2L})-1)|\mathcal{F}_{L+k}], 
\\
I^4_L&=\widehat{\sum_{k=L/2}^{3L/4}}\ \frac{1}{2(L+K)}\tau_{-L-k}\textbf{E}_{\nu_{\beta}}[V^{\prime}(r_{L+k+1})^2p^2_{-L-k-1}\xi^{0,(2L)}_{L+k}\varphi(\mathcal{E}_{-2L,2L})|\mathcal{F}_{L+k}], 
\\
I^5_L&=\widehat{\sum_{k=L/2}^{3L/4}}\ \ \sum_{i=-L-k}^{L+k}\frac{1}{2(L+K)}\tau_{-i}\textbf{E}_{\nu_{\beta}}[V^{\prime}(r_{L+k+1})^2p^2_{-L-k-1}\epsilon^{0,(2L)}_{i} \varphi(\mathcal{E}_{-2L,2L})|\mathcal{F}_{L+k}], 
\\
J^1_L&=\widehat{\sum_{k=L/2}^{3L/4}}\ \ \widehat{\sum_{i=-L-k}^{L+k-1}}\tau_{-i}\textbf{E}_{\nu_{\beta}}[V^{\prime}(r_{L+k+1})^2p^2_{-L-k-1}(\xi^{1,(2L)}_{i}-\xi^{1,(L+k)}_{i})\varphi(\mathcal{E}_{-2L,2L})|\mathcal{F}_{L+k}], 
\\
J^2_L&=\widehat{\sum_{k=L/2}^{3L/4}}\ \ \widehat{\sum_{i=-L-k}^{L+k-1}}\tau_{-i}\{(\xi^{1,(L+k)}_{i}-\xi^1_{i})\textbf{E}_{\nu_{\beta}}[V^{\prime}(r_{L+k+1})^2p^2_{-L-k-1}\varphi(\mathcal{E}_{-2L,2L})|\mathcal{F}_{L+k}] \}, 
\\
J^3_L&=\widehat{\sum_{k=L/2}^{3L/4}}\ \ \widehat{\sum_{i=-L-k}^{L+k-1}}\xi^1(p,r)\tau_{-i}\textbf{E}_{\nu_{\beta}}[V^{\prime}(r_{L+k+1})^2p^2_{-L-k-1}(\varphi(\mathcal{E}_{-2L,2L})-1)|\mathcal{F}_{L+k}],  
\\
J^4_L&=\widehat{\sum_{k=L/2}^{3L/4}}\ \ \widehat{\sum_{i=-L-k}^{L+k-1}}\tau_{-i}\textbf{E}_{\nu_{\beta}}[V^{\prime}(r_{L+k+1})^2p^2_{-L-k-1}\epsilon^{1,(2L)}_{i}\varphi(\mathcal{E}_{-2L,2L})|\mathcal{F}_{L+k}], 
\\
R^1_L&=\widehat{\sum_{k=L/2}^{3L/4}}\tau_{-L-k}\{V^{\prime}(r_{L+k+1})\frac{\partial}{\partial p_{L+k}}g^{(L,k)}\}, 
\\
R^2_L&=\widehat{\sum_{k=L/2}^{3L/4}}\tau_{L+k+1}\{p_{-L-k-1}\frac{\partial}{\partial r_{-L-k}}g^{(L,k)}\}. 
\end{align*}
Here the hat over the sum symbol means that it is in fact an average,
and $\mathcal{E}_{-2L,2L}$ is equal to
$\frac{1}{4L+1}\sum_{i=-2L}^{2L} \mathcal{E}_i$. 

The proof of the theorem will be concluded in the following way. First we show that the middle terms $I^1_L,I^2_L,I^3_L,I^4_L,I^5_L$ and $J^1_L,J^2_L,J^3_L,J^4_L$ tend to zero in $L^2({\nu_{\beta}})$. Then, the proof will be concluded by showing the existence of a subsequence of  $\{-R^1_L + R^2_L\}_{L \geq 1}$ weakly convergent to $cp_0V^{\prime}(r_1)$ with some constant $c$.
\\
For the sake of clarity, the proof is divided in three steps. Before that, let us state two remarks. 
\begin{remark}\label{condconv}
We know that for $m=0,1$, $\textbf{E}_{\nu_{\beta}}[\xi^m|\mathcal{F}_{L}]\xrightarrow{L^2}\xi^m$, \textit{i.e} given $\epsilon>0$ there exist $L_0\in \mathbb{N}$ such that
\[
\textbf{E}_{\nu_{\beta}}[|\xi^m-\xi^{m,(L)}|^2]\leq \epsilon\ \  \textit{if}\ \ L\geq L_0.
\]
Moreover, by the translation invariance we have
\[
\textbf{E}_{\nu_{\beta}}[|\xi_{i}^m-\xi^{m,(L)}_{i}|^2]\leq \epsilon\ \  \textit{if}\ \ [-L_0+i,L_0+i]\subseteq [-L,L].
\]
In fact, given $\tau_{-i}A \in \mathcal{F}_{L}$
\begin{align*}
\int_{A}\xi^{m,(L)}_{i}(\tau_{-i}(p,r))\nu_{\beta}(dpdr)&=\int_{\tau_{-i}(A)}\xi^{m,(L)}_{i}(p,r)\nu_{\beta}(dpdr)\\=\int_{\tau_{-i}(A)}\xi_{i}^m(p,r)\nu_{\beta}(dpdr)
                                       &=\int_{A}\xi_{i}^m(\tau_{-i}(p,r))\nu_{\beta}(dpdr)=\int_{A}\xi^m(p,r)\nu_{\beta}(dpdr).
\end{align*}
In addition, since $\xi^{m,(L)}_{i}(\tau_{-i})\in \mathcal{F}^{L-i}_{-L-i}$ we have 
\[ \xi^{m,(L)}_{i}(\tau_{-i})=\textbf{E}_{\nu_{\beta}}[\xi^m|\mathcal{F}^{L-i}_{-L-i}]\]
and therefore
\[
\textbf{E}_{\nu_{\beta}}[|\xi^m_{i}-\xi^{m,(L)}_{i}|^2]=\textbf{E}_{\nu_{\beta}}[|\xi^m-\xi^{m,(L)}_{i}(\tau_{-i})|^2]\leq \textbf{E}_{\nu_{\beta}}[|\xi^m-\xi^{m,(L_0)}_{0}|^2].
\]
\end{remark}
\begin{remark}\label{stronglaw} Besides a Strong law of large numbers for $(p^2_iV^{\prime}(r_i)^2)_{i\in \mathbb{Z}}$ we have
\[
\textbf{E}_{\nu_{\beta}}\left[\left(\frac{1}{L}\sum_{i=1}^{L}p_i^2V^{\prime}(r_i)^2 -\frac{\phi_{\beta}}{\beta} \right)^2\right]\leq\frac{C_{\beta}}{L}
\]
for some finite constant $C_{\beta}$.
\end{remark}
\textbf{Step 1. The convergence of the middle terms to $0$.}
The convergence to zero as L tends to infinity of $I^1_L$, $I^2_L$ and $I^5_L$ in $L^2({\nu_{\beta}})$ follows from Schwarz inequality, Remark \ref{condconv}, the condition of $\{ \epsilon_L \}$ and the fact that $\varphi$ is a bounded function. \\
Using the symmetry of the measure about exchanges of variables, $I^3_L$ can be rewritten as
\[
\xi^0(p,r)\widehat{\sum_{k=L/2}^{3L/4}}\
\widehat{\sum_{i=-L-k}^{L+k-1}}\textbf{E}_{\nu_{\beta}}[\widehat{\sum_{j=1}^{L-k}}V^{\prime}(r_{L+k+j})^2p^2_{-L-k-j}(\varphi(\mathcal{E}_{-2L,2L})-1)|\mathcal{F}_{L+k}](\tau_{-i}(p,r)) 
\]
and then we decompose it as $I^6_L+\frac{\phi_{\beta}}{\beta}I^7_L$, where $I^6_L$ and $I^7_L$ are respectively
\begin{align*}
\xi^0(p,r) & \widehat{\sum_{k=L/2}^{3L/4}} \
\widehat{\sum_{i=-L-k}^{L+k-1}} \\
& \textbf{E}_{\nu_{\beta}}[\widehat{\sum_{j=1}^{L-k}}\{V^{\prime}(r_{L+k+j})^2p^2_{-L-k-j}-\frac{\phi_{\beta}}{\beta} \}(\varphi(\mathcal{E}_{-2L,2L})-1)|\mathcal{F}_{L+k}](\tau_{-i}(p,r)) 
\end{align*}
and
\[
\xi^0(p,r)\widehat{\sum_{k=L/2}^{3L/4}}\
\widehat{\sum_{i=-L-k}^{L+k-1}}\textbf{E}_{\nu_{\beta}}[\varphi(\mathcal{E}_{-2L,2L})-1|\mathcal{F}_{L+k}](\tau_{-i}(p,r)). 
\]
For the first term, observe that
\[
|I^6_L|^2\leq |\xi^0(p,r)|^2\widehat{\sum_{k=L/2}^{3L/4}}\
\widehat{\sum_{i=-L-k}^{L+k-1}}\textbf{E}_{\nu_{\beta}}\left[\left(\widehat{\sum_{j=1}^{L-k}}\{V^{\prime}(r_{L+k+j})^2p^2_{-L-k-j}-\frac{\phi_{\beta}}{\beta} \}\right)^2\right],
\]
and the expectation inside the last expression is bounded by $\frac{C_{\beta}}{L-k}$, so
\[
||I^6_L||^2_{L^2({\nu_{\beta}})}\leq \frac{C_{\beta}}{L}||\xi^0||^2_{L^2({\nu_{\beta}})}.
\]
For the second term, written explicitly the conditional expectation we see that $|I^7_L|^2$ is  bounded by
\[
|\xi^0(p,r)|^2\widehat{\sum_{k=L/2}^{3L/4}}\
\widehat{\sum_{i=-L-k}^{L+k-1}}\int |\varphi(\frac{1}{4L+1}\sum_{|j|>L+k}\mathcal{E}^{\prime}_{j}+\frac{1}{4L+1}\sum_{|j|\leq L+k}\mathcal{E}_{j+i})-1|^2d\nu_{\beta}.
\]
We rewrite the integral part as
\[
\int |\varphi(\frac{1}{4L+1}\sum_{|j|>L+k}(\mathcal{E}^{\prime}_{j} - E_{\beta})+\frac{1}{4L+1}\sum_{|j|\leq L+k}(\mathcal{E}_{j+i}-E_{\beta})+E_{\beta})-1|^2d\nu_{\beta}.
\]
Using the fact that $\varphi$ is a Lipschitz positive function bounded by $1$ such that $\varphi(E_{\beta})=1$, we obtain that $|I^7_L|^2$ is bounded from above by
\[
|\xi^0(p,r)|^2\widehat{\sum_{k=L/2}^{3L/4}}\
\widehat{\sum_{i=-L-k}^{L+k-1}}1\wedge\int|\frac{1}{4L+1}\sum_{|j|>L+k} (\mathcal{E}^{\prime}_{j}-E_{\beta})+\frac{1}{4L+1}\sum_{|j|\leq L+k}(\mathcal{E}_{j+i}-E_{\beta})|^2d\nu_{\beta}
\]
where $a \wedge b$ denote the minimum of $\{a,b\}$. So, taking expectation and using the Strong law of large numbers together with the dominated convergence theorem, the convergence to zero as L tends to infinity of $I^3_L$ in $L^2(\nu_{\beta})$ is proved.

Same arguments can be applied for $J^1_L$, $J^2_L$, $J^3_L$ and $J^4_L$.
For $I^4_L$, we can bound the $L^2$-norm of the term from above by $\frac{C_{\beta}}{L}||\xi^0||^2_{L^2(\nu_{\beta})}$ for some constant $C_{\beta}$.

\textbf{Step 2. The uniform bound of the $L^2(\nu_{\beta})$ norms of the boundary terms.} 

Remember that $R^1_L$ is defined as
\begin{align*}
\widehat{\sum_{k=L/2}^{3L/4}} & \frac{1}{2(L+k)} \tau_{-L-k}\{V^{\prime}(r_{L+k+1})\textbf{E}_{\nu_{\beta}}[p^2_{-L-k-1}V^{\prime}(r_{L+k+1})^2\frac{\partial}{\partial p_{L+k}}g^{(2L)}|\mathcal{F}_{L+k}]\}\\ 
= -\widehat{\sum_{k=L/2}^{3L/4}} & \frac{1}{2(L+k)} \\
& \tau_{-L-k}\{V^{\prime}(r_{L+k+1})\textbf{E}_{\nu_{\beta}}[p^2_{-L-k-1}V^{\prime}(r_{L+k+1})Y_{L+k,L+k+1}g^{(2L)}|\mathcal{F}_{L+k}]\}\\ 
+\widehat{\sum_{k=L/2}^{3L/4}} & \frac{1}{2(L+k)} \\
& \tau_{-L-k} \{p_{L+k}V^{\prime}(r_{L+k+1})\textbf{E}_{\nu_{\beta}}[p^2_{-L-k-1}V^{\prime}(r_{L+k+1})\frac{\partial}{\partial r_{L+k+1}}g^{(2L)}|\mathcal{F}_{L+k}]\}.
\end{align*}
By Schwarz inequality and (\ref{system}), we can see that the $L^2(\nu_{\beta})$ norm of the first term in the right hand side of the last equality is bounded by $\frac{C_{\beta}}{L}||\xi ^1||_{L^2(\nu_{\beta})}$ for some constant $C_{\beta}$. After an integration by parts, the second term can be written as 
\begin{equation}
\label{secondterm}
\begin{split}
& \widehat{\sum_{k=L/2}^{3L/4}} \frac{1}{2(L+k)} \\
&\tau_{-L-k}\{p_{L+k}V^{\prime}(r_{L+k+1})\textbf{E}_{\nu_{\beta}}[p^2_{-L-k-1}\big(\beta V^{\prime}(r_{L+k+1})^2-V^{\prime \prime}(r_{L+k+1})\big)g^{(2L)}|\mathcal{F}_{L+k}]\}.
\end{split}
\end{equation}
Using the symmetry of the measure again, the conditional expectation appearing in the last expression can be rewritten as
\[
\textbf{E}_{\nu_{\beta}}[p^2_{-L-k-1}\widehat{\sum_{j=L+k+1}^{2L}}(\beta V^{\prime}(r_j)^2-V^{\prime \prime}(r_j))(g^{(2L)}\circ\pi^{j,L+k+1}_r)|\mathcal{F}_{L+k}]\;,
\]
where $\pi^{j,L+k+1}_r$ stands for the exchange operator of $r_j$ and $r_{L+k+1}$. After that, we decompose the last expression as the sum of the following two terms,
\[
\textbf{E}_{\nu_{\beta}}[p^2_{-L-k-1}\widehat{\sum_{j=L+k+1}^{2L}}(\beta V^{\prime}(r_j)^2-V^{\prime \prime}(r_j))g^{(2L)}|\mathcal{F}_{L+k}],
\]
and
\[
\textbf{E}_{\nu_{\beta}}[p^2_{-L-k-1}\widehat{\sum_{j=L+k+1}^{2L}}(\beta V^{\prime}(r_j)^2 -V^{\prime \prime}(r_j))(g^{(2L)}\circ\pi^{j,L+k+1}_r-g^{(2L)})|\mathcal{F}_{L+k}].
\]
The square of the last expressions are respectively bounded from above by
\[
C_{\beta}L^{-1}\textbf{E}_{\nu_{\beta}}[(g^{(2L)})^2|\mathcal{F}_{L+k}],
\quad
C_{\beta} \textbf{E}_{\nu_{\beta}}[\widehat{\sum_{j=L+k+1}^{2L}}(g^{(2L)}\circ\pi^{j,L+k+1}_r-g^{(2L)})^2|\mathcal{F}_{L+k}]
\]
for some constant $C_{\beta}$. Using Schwarz inequality we can see that the square of each term of the sum is respectively bounded from above by 
\begin{equation}
\frac{C_{\beta}}{L^3}\textbf{E}_{\nu_{\beta}}\big[\big(\widehat{\sum_{k=L/2}^{3L/4}}p^2_{L+k}\big)(g^{(2L)})^2\big],\label{eq:27}
\end{equation}
and
\begin{align*}
& \frac{C'_{\beta}}{L^2}\widehat{\sum_{k=L/2}^{3L/4}}\textbf{E}_{\nu_{\beta}}\big[p^2_{L+k}   \widehat{\sum_{j=L+k+1}^{2L}}(g^{(2L)}\circ\pi^{j,L+k+1}_r-g^{(2L)})^2\big] \\
& \le \frac{C'_{\beta}}{L^2}\widehat{\sum_{k=L/2}^{3L/4}}\textbf{E}_{\nu_{\beta}}\big[p^2_{L+k}
   \widehat{\sum_{j=L+k+1}^{2L}}2\{j-(L+k+1)\}\sum_{i=L+k+1}^{j-1}(g^{(2L)}\circ\pi^{i,i+1}_r-g^{(2L)})^2\big] \\
& \le \frac{C'_{\beta}}{L}\widehat{\sum_{k=L/2}^{3L/4}}\textbf{E}_{\nu_{\beta}} \big[p^2_{L+k}
   \sum_{i=3L/2+1}^{2L}(g^{(2L)}\circ\pi^{i,i+1}_r-g^{(2L)})^2\big] 
\end{align*}
for some constants $C_{\beta}$ and $C'_{\beta}$.
One can now estimate $\widehat{\sum_{k=L/2}^{3L/4}}p^2_{L+k}$
uniformly because of the cutoff. Using the spectral gap estimate
(\ref{eq:18}) proved in Section \ref{sec:spectral-gap-bound}, we can
bound \eqref{eq:27} by a constant. 

Finally, we state that we can bound the term $\textbf{E}_{\nu_{\beta}} \big[(g^{(2L)}\circ\pi^{i,i+1}_r-g^{(2L)})^2\big]$ by the Dirichlet form of $g^{(2L)}$ which concludes the proof.
\begin{proposition}
There exists some constant $C$ such that for every smooth function $f:\Omega \to \R$,
\[
\textbf{E}_{\nu_{\beta}} \big[(f \circ\pi^{i,i+1}_r-f)^2\big] \le C \{ \textbf{E}_{\nu_{\beta}}\big[(X_if)^2\big]+\textbf{E}_{\nu_{\beta}} \big[(Y_{i,i+1}f)^2\big] \}.
\]
\end{proposition}
\begin{proof}
The change of variables and simple computations conclude the proof.
\end{proof}

\textbf{Step 3. The existence of a weakly convergent subsequence of  $\{R^1_L\}_{L \geq 1}$.}
Firstly, observe that the expression (\ref{secondterm}) is equal to 
\[
p_0V^{\prime}(r_1)h^{1}_L(p_0,r_0, \dots, p_{-7L/2},r_{-7L/2})
\] 
where
\begin{equation*}
h^{1}_L=\widehat{\sum_{k=L/2}^{3L/4}}\frac{1}{2(L+k)}\tau^{-L-k}\textbf{E}_{\nu_{\beta}}[p^2_{-L-k-1}(\beta V^{\prime}(r_{L+k+1})^2-V^{\prime \prime}(r_{L+k+1}))g^{(2L)}|\mathcal{F}_{L+k}].
\end{equation*}
On the other hand, we had proved in \textbf{Step 2} that $\{p_0V^{\prime}(r_1)h^{1}_L\}_{L \geq 1}$ is bounded in $L^2(\nu_{\beta})$, therefore it contains a weakly convergent subsequence $\{p_0V^{\prime}(r_1)h_{L'}\}_{L'}$. We can conclude in a similar way that $\{h^{1}_L\}_{L \geq 1}$ is bounded in $L^2(\nu_{\beta})$, therefore $\{h^{1}_{L'}\}_{L'}$ contains a weakly convergent subsequence, whose limit will be denoted by $h$.
It is easy to see that
\[
||X_{i}h^{1}_L||_{L^2(\nu_{\beta})} \leq \frac{C}{L}||\xi^0||_{L^2(\nu_{\beta})} \quad \quad \textit{for} \quad \quad i \in \{0,-1,-2,\cdots\}
\]
and
\[
||Y_{i,i+1}h^{1}_L||_{L^2(\nu_{\beta})} \leq \frac{C}{L}||\xi^1||_{L^2(\nu_{\beta})} \quad \quad \textit{for} \quad \quad \{i,i+1\} \subseteq \{0,-1,-2,\cdots\}
\]
which implies that $X_{i}h=0$ for $i \in \{0,-1,-2,\cdots\}$ and $Y_{i,i+1} h=0$ for $\{i,i+1\} \subseteq \{0,-1,-2,\cdots\}$.
Since the function $h$ depends only on $\{p_0, r_0, p_{-1}, r_{-1}, p_{-2}, r_{-2} \cdots\}$ one can show that $h$ is a constant function, let's say $c$. Taking suitable test functions, we can conclude that in fact $\{p_0V^{\prime}(r_1)h^{1}_{L'}\}_{L'}$ converges weakly to $cp_0V^{\prime}(r_1)$.
This proves that for every weakly convergent subsequence of $\{R^1_L\}_{L \geq 1}$ there exist a constant $c$ such that the limit is $cp_0V^{\prime}(r_1)$. Exactly the same can be said about $\{R^2_L\}_{L \geq 1}$.
\end{proof}

\begin{remark}
Observe that the roles of the vector fields $X_0$ and $Y_{0,1}$ are
symmetric, in the sense that changing the definition of the energy of
the particle $i$ to $\mathcal E_i = p_i^2/2 + V(r_{i+1})$ their actions
in the boundary terms in the above approximation are exchanged. The
space of closed forms does not depend on this choice of the definition
of the energy $\mathcal E_i$, so we also have the equivalent
characterization of the closed forms:
  \begin{equation}
  \label{eq:12al}
  \mathfrak{H}_c \ = \ \overline{\mathcal{B} + \{( p_0 V^{\prime}(r_0), 0)\}}.
\end{equation}
This imply that, defining by $\xi_F = (X_0 \Gamma_F, Y_{0,1}
\Gamma_F)$,
 a closed form $\xi$ can be approximated by $\xi_F + c_0 (p_0
 V^{\prime}(r_0), 0)$ and by $\xi_G + c_1 (0, p_0 V^{\prime}(r_1))$,
 then $c_0 = -c_1 = c$ and $F-G = - c \frac{p_0^2}{2}$.
\end{remark}

\section{Diffusion coefficient}\label{sec-diff-coeff}

In this section, we describe the diffusion coefficient in several variational formulas and prove the second statement of Lemma \ref{lem:CLT}. From Corollary
\ref{cor:decompositionch3}, there exists a unique number
$\kappa_\beta$
such that 
\begin{displaymath}
W_{0,1} + \kappa_\beta (p_1^2- p_0^2)  \in \overline{L\Cy} \quad
\text{in} \quad \mathcal{H}_{-1}.  
\end{displaymath}

Our purpose now is to obtain a more explicit formula for $\kappa_\beta$. To
do this, we follow the argument in \cite{LOY1}.    
\begin{lemma}
We have
\begin{displaymath}
\mathcal{H}_{-1} = \overline{L\Cy}|_{\mathcal{N}} \oplus
\{W_{0,1}\} = \overline{L^{*}\Cy}|_{\mathcal{N}} \oplus
\{W^{*}_{0,1}\} 
\end{displaymath}
where $W^{*}_{0,1}:=W^S_{0,1}-W^A_{0,1}$ and $L^{*}=S-A $. 
\end{lemma}
\begin{proof}
We shall prove the first decomposition since the same arguments apply
to the second one. Because we have already proved in Lemma
\ref{lem:decompositionch3} that $\overline{L\Cy}|_{\mathcal{N}}$ has a
one-dimensional complementary subspace in $\mathcal{H}_{-1}$, 
 it is sufficient to show that $\mathcal{H}_{-1}$ is generated by
 $\overline{L\Cy}$ and the current.
 Let $h \in \mathcal{H}_{-1}$ so that $\ll h, W_{0,1} \gg_{-1}=0$ and
 $\ll h, Lg \gg_{-1}=0$ for all $g \in \Cy$. 
 By Proposition \ref{prop:directch3}, $h=\lim_{k \to \infty}
 (aW^S_{0,1}+S h_k)$ in $\mathcal{H}_{-1}$ for some $a \in \R$ and
 $h_k \in \Cy$. 
In particular, 
\begin{equation*}
  \begin{split}
    \triple h \triple_{-1}^2 = \lim_{k \to \infty} \ll aW^S_{0,1} + Sh_k,
    aW^S_{0,1} + Sh_k \gg_{-1}\\
    =\lim_{k \to \infty} \ll aW^S_{0,1} +
    Sh_k, aW_{0,1}+ Lh_k \gg_{-1}\label{eq:13}
  \end{split}
\end{equation*}
since $\ll aW^S_{0,1} + Sh_k, a W^A_{0,1}+ A h_k \gg_{-1}=0$ by Lemma
\ref{lem:anti3ch3}. 
 On the other hand, by assumption $\ll h, aW_{0,1}+Lh_k
 \gg_{-1}=0$ for all $k$. 
 Also, by Proposition \ref{prop:aboundch3}, 
$$
\sup_k \triple aW_{0,1}+Lh_k \triple_{-1}^2 \le 2a^2 \triple
W_{0,1} \triple_{-1}^2 + 2 (C+1) \sup_k \triple
Sh_k \triple_{-1}^2 :=C_h
$$ 
is finite. Therefore, 
\begin{equation*}
  \begin{split}
    \triple h \triple_{-1}^2 = \lim_{k \to \infty} \ll
    aW^S_{0,1}+Sh_k, aW_{0,1}+Lh_k
    \gg_{-1}\\
    =\lim_{k \to \infty} \ll
    aW^S_{0,1}+Sh_k-h, aW_{0,1} +Lh_k
    \gg_{-1} \\
    \le \limsup_{k \to \infty} C_h \triple
      aW^S_{0,1}+S h_k-h \triple_{-1} =0 .
  \end{split}
\end{equation*}
 This concludes the proof.
\end{proof}
Now, we can define bounded linear operators $T:\mathcal{H}_{-1} \to
\mathcal{H}_{-1}$   \\and $T^*:\mathcal{H}_{-1} \to \mathcal{H}_{-1}$ as
\begin{displaymath}
  \begin{split}
    T(aW_{0,1}+ Lf):=aW^S_{0,1}+Sf, \\
    T^*(aW^*_{0,1}+L^{*}f):=aW^S_{0,1}+Sf \ .
  \end{split}
\end{displaymath}
Since 
$$
\triple aW_{0,1} + Lf \triple_{-1}^2=
\triple aW^*_{0,1} + L^{*}f \triple_{-1}^2 =
\triple aW^S_{0,1} + Sf \triple_{-1}^2
+\triple a W^A_{0,1} + Af \triple_{-1}^2 .
$$
 we can easily show that $T^*$ is the adjoint operator of $T$ and also we have the relations
\begin{displaymath}
\ll T(p_1^2-p_0^2), W^{*}_{0,1} \gg_{-1} = \ll T^*(p_1^2-p_0^2), W_{0,1} \gg_{-1} = -\frac{1}{\beta^2},
\end{displaymath}
and
\begin{displaymath}
\ll T(p_1^2-p_0^2), L^{*}f \gg_{-1} = \ll T^*(p_1^2-p_0^2), Lf \gg_{-1}=0
\end{displaymath}
for all $f \in \Cy$. 
In particular, 
$$
\mathcal{H}_{-1} = \overline{L^{*}\Cy}|_{\mathcal{N}} \oplus
  \{T(p_1^2-p_0^2)\}
$$ 
and there exists a unique number $Q_\beta$ such that 
\begin{displaymath}
W^*_{0,1} + Q_\beta T(p_1^2-p_0^2) \in \overline{L^{*}\Cy} \quad
\text{in} \quad \mathcal{H}_{-1}.  
\end{displaymath}
It will turn out later that $Q_\beta = \kappa_\beta$.

\begin{lemma}\label{prop:qvariationalch3}
\begin{equation}\label{eq:qvariationalch3}
Q_\beta = \frac{1}{\beta^2 \triple T(p_1^2-p_0^2) \triple_{-1}^2} 
=\beta^2 \inf_{f \in \Cy}\triple W^{*}_{0,1}-L^{*}f \triple_{-1}^2. 
\end{equation}
\end{lemma}
\begin{proof}
First identity follows from the fact that 
\begin{displaymath}
\ll T(p_1^2-p_0^2), W^*_{0,1} + Q_\beta T(p_1^2-p_0^2) \gg_{-1} 
=-\frac{1}{\beta^2} + Q_\beta \triple T(p_1^2-p_0^2)
\triple_{-1}^2 =0. 
\end{displaymath}
Second identity is obtained by the expression
\begin{displaymath}
\inf_{f \in \Cy}\triple W^*_{0,1} + Q_\beta T(p_1^2-p_0^2) -
L^{*}f \triple_{-1}  =0
\end{displaymath}
since 
\begin{align*}
\inf_{f \in \Cy}  & \triple W^*_{0,1} + Q_\beta T(p_1^2-p_0^2) 
- L^{*}f \triple_{-1}^2 \\
&= \inf_{f \in \Cy}  \triple W^*_{0,1} - L^{*}f \triple_{-1}^2  -
\frac{2Q_\beta}{\beta^2} + Q_\beta^2
\triple T(p_1^2-p_0^2) \triple_{-1}^2 \\
&= \inf_{f \in \Cy}  \triple W^*_{0,1} - L^{*}f \triple_{-1}^2  -
\frac{2Q_\beta}{\beta^2} +  \frac{Q_\beta}{\beta^2}.
\end{align*}
\end{proof}
By a simple computation, we can show that 
$\ll Tg, g \gg_{-1} =\ll Tg, Tg \gg_{-1}$ for all 
$g \in \mathcal{H}_{-1}$, and therefore 
$(p_1^2-p_0^2)-T(p_1^2-p_0^2) \in \overline{L^{*}\Cy_0}$ 
since $(p_1^2-p_0^2)-T(p_1^2-p_0^2)$ is orthogonal to 
$T(p_1^2-p_0^2)$. 
Then we obtain the variational formula for 
$\triple T(p_1^2-p_0^2) \triple_{-1}^2$:   
\begin{lemma}
\begin{equation}\label{eq:qvariational2ch3}
\triple T(p_1^2-p_0^2) \triple_{-1}^2
=\inf_{f \in \Cy}\triple p_1^2-p_0^2-L^{*}f \triple_{-1}^2. 
\end{equation}
\end{lemma}
\begin{proof}
By the similar argument with the proof of Proposition \ref{prop:qvariationalch3}, we have 
\begin{displaymath}
\inf_{f \in \Cy}  \triple p_1^2-p_0^2 - T(p_1^2-p_0^2) - L^{*}f
\triple_{-1}^2   =0
\end{displaymath}
and 
\begin{align*}
& \inf_{f \in \Cy}   \triple p_1^2-p_0^2 - T(p_1^2-p_0^2) - L^{*}f
\triple_{-1}^2 \\
&= \inf_{f \in \Cy}  \triple p_1^2-p_0^2 - L^{*}f  \triple_{-1}^2
- \triple T(p_1^2-p_0^2) \triple_{-1}^2
\end{align*}
which concludes the proof.
\end{proof}

\begin{proposition}
\begin{equation}\label{eq:variationalch3}
\kappa_\beta= \beta^2 \inf_{f \in \Cy}\triple W^{*}_{0,1}-L^{*}f
 \triple_{-1}^2 = \frac{1}{\beta^2 \inf_{f \in \Cy }\triple
  p_1^2-p_0^2-L^{*}f  \triple_{-1}^2}. 
\end{equation}
\end{proposition}
\begin{proof}
By the definition, $W_{0,1} + \kappa_\beta (p_1^2-p_0^2)  \in \overline{L\Cy}$ and therefore 
\begin{displaymath}
\ll W_{0,1} + \kappa_\beta (p_1^2-p_0^2),
T^*(p_1^2-p_0^2) \gg_{-1}=  -\frac{1}{\beta^2} +
\kappa_\beta\triple T(p_1^2-p_0^2)   \triple_{-1}^2 =0.
\end{displaymath}
Then, $\kappa_\beta=Q_\beta$ follows and we obtain two variational formulas
from (\ref{eq:qvariationalch3}) and (\ref{eq:qvariational2ch3}).  
\end{proof}

\begin{proposition}\label{prop:CLT2}
For any sequence $F_K$ in $\Cy$ such that 
$$
\lim_{K \to \infty} \triple
W_{0,1}+\kappa_\beta(p_1^2-p_0^2)- LF_K
 \triple_{-1}=0,
$$
we have
\begin{equation*}
\lim_{K \to \infty} [  \frac{ \gamma}{2} \langle ( p_0V'(r_1)- Y_{0,1}\Gamma_{F_K})^2 \rangle +  \frac{ \gamma}{2} \langle  ( X_{0}\Gamma_{F_K} )^2 \rangle ] = \frac{\kappa_\beta}{\beta^2}.
\end{equation*}
\end{proposition}
\begin{proof}
By the assumption, 
$$
\lim_{K \to \infty} \triple T \{
W_{0,1}+\kappa_\beta(p_1^2-p_0^2)- LF_K \}
\triple_{-1} =0
$$ 
and therefore 
\begin{displaymath}
\lim_{K \to \infty}\triple W^S_{0,1} -SF_K \triple_{-1}^2 =
\kappa_\beta^2\triple T(p_1^2-p_0^2)  \triple_{-1}^2. 
\end{displaymath}
Then, since 
$$
\kappa_\beta=Q_\beta =
\frac{1}{ \beta^2 \triple T(p_1^2-p_0^2)  \triple_{-1}^2} 
$$ 
and 
$$
\triple W^S_{0,1} -SF_K  \triple_{-1}^2= 
\frac{ \gamma}{2} \langle ( p_0V'(r_1)- Y_{0,1}\Gamma_{F_K})^2
\rangle +  \frac{ \gamma}{2} \langle  ( X_{0}\Gamma_{F_K} )^2
\rangle,
$$
we complete the proof.  
\end{proof}

\section{Detailed estimates of the diffusion coefficient}
\label{sec:upperboundch3}

In this section, we give some detailed estimates of the diffusion coefficient as a function of $\gamma$. Note that they are not necessary to prove our main theorem.

First, we rewrite the variational formula for the diffusion coefficient given by the terms of the norm of $\mathcal{H}_{-1}$ in a tractable way.

Observe that $\Cy$ is divided into two orthogonal spaces
$\mathbb{L}_e$ and $\mathbb{L}_o$ where $\mathbb{L}_e$ is the set of even functions in $p$ and $\mathbb{L}_o$ is the set of odd functions in $p$. More precisely, for $f \in \Cy$, $ f \in \mathbb{L}_e$ if and only if $f(p,r)=f(-p,r)$ and $ f \in \mathbb{L}_o$ if and only if $f(p,r)=-f(-p,r)$  where $(-p)_i=-p_i$ for all $i$. 

Consider two subspaces of $\mathcal{H}_{-1}$ defined as  
$\mathcal{H}^{e}_{-1} := \overline { S \mathbb{L}_e} | _{\mathcal{N}}
\oplus \{ W^S_{0,1} \}$  and $\mathcal{H}^{o}_{-1} :=
\overline{S\mathbb{L}_o} |_{\mathcal{N}} $. 

\begin{lemma}
We have
\begin{displaymath}
\mathcal{H}_{-1} =\mathcal{H}^{e}_{-1} \oplus \mathcal{H}^{o}_{-1}
\end{displaymath}
and they are orthogonal to each other. Moreover, $W^A_{0,1} \in
\mathcal{H}^{o}_{-1}$, $Af \in\mathcal{H}^{o}_{-1}$  if $f \in
\mathbb{L}_e$  and $Af \in \mathcal{H}^{e}_{-1}$ if $f \in \mathbb{L}_o$.
\end{lemma}

\begin{proof}
Straightforward.
\end{proof}

\begin{proposition}
\begin{equation} \label{eq:variational2ch3}
\begin{split}
\kappa_\beta & = \beta^2 \inf_{f \in \mathbb{L}_e }\sup_{g \in
  \mathbb{L}_o}  \Big\{  \gamma [  \frac{1}{2} \langle ( p_0V'(r_1)-
Y_{0,1}\Gamma_f)^2 \rangle +  \frac{1}{2} \langle  (
X_{0}\Gamma_f )^2 \rangle ]  \\  
+ &  2  \langle (W^{A}_{0,1}- Af) \Gamma_g  \rangle -  \gamma
[ \frac{1}{2} \langle (Y_{0,1}\Gamma_g)^2 \rangle +
\frac{1}{2} \langle  ( X_{0}\Gamma_g)^2 \rangle ] \Big\}. 
\end{split}
\end{equation}
\end{proposition}
\begin{proof}
We can rewrite the first variational formula for $\kappa_\beta$ in
(\ref{eq:variationalch3}) as 
\begin{align*}
& \beta^2 \inf_{f \in \Cy}\{\triple W^{S}_{0,1} - Sf
\triple_{-1}^2 +\triple W^{A}_{0,1}- Af  \triple_{-1}^2\}   \\
&=\beta^2  \inf_{f_e \in \mathbb{L}_e}\inf_{f_o \in\mathbb{L}_o}\{
\triple W^{S}_{0,1} - Sf_e\triple_{-1}^2 + \triple Sf_o\triple_{-1}^2
\nonumber +\triple W^{A}_{0,1} - Af_e\triple_{-1}^2 + \triple
Af_o\triple_{-1}^2 \}  \nonumber \\ 
 &=\beta^2  \inf_{f \in \mathbb{L}_e} \{ \triple W^{S}_{0,1} -
 Sf\triple_{-1}^2 \nonumber +\triple W^{A}_{0,1} - Af\triple_{-1}^2  \}  \nonumber \\
 &= \beta^2  \inf_{f \in \mathbb{L}_e }\sup_{g \in  \mathbb{L}_o} \{
\triple W^{S}_{0,1} - Sf\triple_{-1}^2
- 2 \ll W^A_{0,1}-Af, Sg \gg_{-1} -\triple Sg \triple_{-1}^2\}
\nonumber  \\  
&= \beta^2  \inf_{f \in \mathbb{L}_e }\sup_{g \in  \mathbb{L}_o} \{
\gamma [  \frac{1}{2} \langle ( p_0V'(r_1)- Y_{0,1}\Gamma_f)^2 \rangle +  \frac{1}{2} \langle  ( X_{0}\Gamma_f )^2 \rangle
] +  2  \langle (W^{A}_{0,1}- Af) \Gamma_g  \rangle
\nonumber  \\ 
& \quad \quad - \gamma  [ \frac{1}{2} \langle (Y_{0,1}\Gamma_g)^2
\rangle + \frac{1}{2} \langle  ( X_{0}\Gamma_g)^2 \rangle ] \} \nonumber. 
\end{align*}
\end{proof}

\begin{proposition}\label{prop:upperboundch3}
\begin{displaymath}
\kappa_\beta \le \frac{\gamma}{4}\langle V''(r_0) \rangle+\frac{3}{4\gamma}.
\end{displaymath}
\end{proposition}
\begin{proof}
Take $f=-\frac{p_0^2}{4}$ in the variational formula (\ref{eq:variational2ch3}), then we have
\begin{align*}
\kappa_\beta & \le \beta^2 \sup_{g \in  \mathbb{L}_o} \{   \frac{\gamma}{4} \langle  p_0 ^2 V'(r_0) ^2 \rangle + 2  \langle \{ W^{A}_{0,1}+A(\frac{p_0^2}{4}) \} \Gamma_g  \rangle -  \gamma [\frac{1}{2} \langle (Y_{0,1}\Gamma_g)^2 \rangle + \frac{1}{2} \langle  ( X_{0}\Gamma_g)^2 \rangle ] \} \\
& =  \frac{\gamma}{4}\langle V''(r_0) \rangle + \frac{\beta^2}{\gamma}  \sup_{g \in  \mathbb{L}_o} \{  2  \langle (W^{A}_{0,1}+A(\frac{p_0^2}{4})) \Gamma_g  \rangle -  \frac{1}{2} \langle (Y_{0,1}\Gamma_g)^2 \rangle - \frac{1}{2} \langle  ( X_{0}\Gamma_g)^2 \rangle \}.
\end{align*}
Since $ W^A_{0,1}= Y_{0,1}(\frac{p_0^2}{2})$,
\begin{align*}
\sup_{g \in  \mathbb{L}_o} & \{  2  (W^{A}_{0,1}+A(\frac{p_0^2}{4})) \Gamma_g  \rangle -  \frac{1}{2} \langle (Y_{0,1}\Gamma_g)^2 \rangle - \frac{1}{2} \langle  ( X_{0}\Gamma_g)^2 \rangle \} \\
&= \sup_{g \in  \mathbb{L}_o} \{ - \frac{1}{2} \langle p_0^2, X_{0} \Gamma_g \rangle  - \frac{1}{2} \langle p_0^2, Y_{0,1} \Gamma_g \rangle- \frac{1}{2} \langle (Y_{0,1} \Ga_g)^2 \rangle - \frac{1}{2} \langle  ( X_{0}\Gamma_g)^2 \rangle \} \\
&= \sup_{g \in  \mathbb{L}_o} \{ - \frac{1}{2} \langle(X_{0} \Ga_g + \frac{p_0^2}{2})^2 \rangle - \frac{1}{2} \langle(Y_{0,1} \Ga_g + \frac{p_0^2}{2})^2  \rangle + \frac{1}{4} \langle p_0^4 \rangle \}  \le  \frac{1}{4} \langle p_0^4 \rangle.
\end{align*}
\end{proof}
\begin{proposition}
\begin{displaymath}
\kappa_\beta \ge \frac{\gamma }{4\beta \langle r_0^2 \rangle }.
\end{displaymath}
\end{proposition}
\begin{proof}
By the variational formula (\ref{eq:variational2ch3})
\begin{align*}
\kappa_\beta & \ge \gamma \beta^2 \inf_{f \in \mathbb{L}_e } \{ [  \frac{1}{2} \langle ( p_0V'(r_1)+ Y_{0,1}\Gamma_f)^2 \rangle +  \frac{1}{2} \langle  ( X_{0}\Gamma_f )^2 \rangle ] \}.
\end{align*}
Since $\frac{1}{\beta^2} = \langle p_0V'(r_1),  p_0 r_1 \rangle$ and $ \langle p_0 r_0, X_0(\Gamma_f )  \rangle -  \langle p_0 r_1, Y_{0,1}(\Gamma_f )  \rangle = \langle V'(r_0)r_0-V'(r_1)r_1, \Gamma_f \rangle = 0$ for any $f \in \mathbb{L}_e$, we have
\[
\frac{1}{\beta^2} = \langle p_0V'(r_1) - Y_{0,1}(\Gamma_f ),  p_0 r_1 \rangle +  \langle p_0 r_0, X_0(\Gamma_f )  \rangle
\]
for any $f \in \mathbb{L}_0$. Then, by Schwarz inequality,
\begin{align*}
\frac{1}{\beta^4} & \le \inf_{f \in \mathbb{L}_0 } \langle (p_0V'(r_1) - Y_{0,1}(\Gamma_f ))^2 + (X_0(\Gamma_f ))^2 \rangle   \langle (p_0 r_1)^2+  (p_0 r_0)^2\rangle \\
& = \frac{2}{\beta} \langle  r_0^2  \rangle  \inf_{f \in \mathbb{L}_0 }\langle (p_0V'(r_1) - Y_{0,1}(\Gamma_f ))^2 + (X_0(\Gamma_f ))^2 \rangle.
\end{align*}
\end{proof}

\begin{remark}\label{rmk-harmonic}
For the harmonic case with $V(r)=\frac{r^2}{2}$, we have an
explicit fluctuation-dissipation given by
\begin{equation}
  \label{eq:26}
  \begin{split}
    W^A_{0,1} + W^S_{0,1} &= -p_0 r_1 + \frac{\gamma}2 (p_0^2 - r_1^2) \\
    &= -\nabla \left[ \left(\frac 1{6\gamma} + \frac
        {\gamma}{4}\right) p_0^2 + \frac 12 r_0 r_1 \right] + L\left(
      \frac 1{6\gamma} (p_0 + p_1) r_1 + \frac{r_1^2}{4} \right)
  \end{split}
\end{equation}
i.e. the diffusion coefficient is given by 
$D_\beta=\frac{\gamma}{4}+\frac{1}{6\gamma}$ which does not depend on $\beta$.
\end{remark}

\section{Spectral gap}\label{sec:spectral-gap-bound}
In this section, we prove the spectral gap estimates for the process of finite oscillators without the periodic boundary condition, which is used in the proof of Theorem \ref{thm:closedformch3} in Section \ref{sec:closedformch2}. We use the following notation:
\begin{equation*}
E_{\nu_{L,E}} [ \ \cdot \ ] :=E_{\nu_{\beta}^L}[\ \cdot \ | \frac{1}{L}\sum_{i=1}^L(\frac{p_i^2}{2}+V(r_i))=E].
\end{equation*}
Recall that we assume that $0 < \delta_{-} \le V''(r) \le \delta_{+} < \infty$. Then it is easy to see that $V$ satisfies
\[
0 < d_- \le \Big| \frac{\sqrt{2V(r)}}{V^{\prime}(r)} \Big| \le d_+ < \infty
\]
for all $r \in \R \setminus \{0\}$ where $d_-=\frac{\sqrt{\delta_-}}{\delta_+}$ and $d_+=\frac{\sqrt{\delta_+}}{\delta_-}$. 
Under these assumptions, we can operate the change of variables $(p, r) \to (\E, \theta)$ as $\sqrt{\E} \cos\theta=\frac {p}{\sqrt{2}}$ and $\sqrt{\E} \sin\theta= sgn(r) \sqrt{V(r)}$, and we obtain that
\begin{equation*}
  \int_{\R^2} f(p,r) d \nu_{\beta}^1 = \frac{1}{\sqrt{2\pi \beta^{-1}}Z_{\beta}} \int_0^\infty \int_0^{2\pi} \tilde{f}(\E, \theta) e^{-\beta \E} q(\E,\theta) d\E d\theta
\end{equation*}
where $q(\E,\theta)=|\frac{\sqrt{2V(r(\E, \theta))}}{V^{\prime}(r(\E, \theta))}|$, which satisfies $ d_- \le q(\E,\theta) \le d_+$ for all $\E$ and $\theta$.
Here, $\tilde{f}(\E,\theta):=f(p(\E,\theta),r(\E,\theta))$.

Let $h_{\beta}(x) dx$ be the probability distribution on $\R_+$ of 
$p^2/2 + V(r)$  under $d \nu_{\beta}^1$, i.e. 
\begin{equation*}
  \int_{\R^2} g(p^2/2 + V(r)) d \nu_{\beta}^1 = \int_0^\infty g(x) h_{\beta}(x) dx
\end{equation*}
for any $g: \R_+ \to \R$. Then, since
$h_{\beta}(x)= \frac{1}{\sqrt{2\pi \beta^{-1}}Z_{\beta}} \int_0^{2\pi} e^{-\beta x } q(x,\theta) d\theta$, we obtain
 \[
\frac{\delta_-}{\delta_+}\beta e^{- \beta x} \le h_{\beta}(x) \le \frac{\delta_+}{\delta_-} \beta e^{- \beta x}
\]
for all $x>0$.

With these notations, we prepare two lemmas before we state the main result of this section.

\begin{lemma} \label{lem:2ch3}
There exists a positive constant $C$ such that
  \begin{equation*}
    E_{\nu_{1,E}}[(f-E_{\nu_{1,E}}[f])^2] 
    \le C E_{\nu_{1,E}}[(X_1 f)^2] 
  \end{equation*}
for every $E>0$, and every smooth function $f$.
\end{lemma}
\begin{proof}
By simple computations with the change of variable,
\begin{equation*}
  E_{\nu_{1,E}}[(f-E_{\nu_{1,E}}[f])^2] =  \frac{ \int_0^{2\pi} (\tilde{f}(E,\theta)-E_{\nu_{1,E}}[f])^2 q(E,\theta) d\theta}{\int_0^{2\pi}  q(E,\theta) d\theta}
\end{equation*}
and 
\begin{equation*}
  E_{\nu_{1,E}}[(X_1 f)^2] =  \frac{ \int_0^{2\pi} \{ q(E,\theta)^{-1}\partial_{\theta} \tilde{f}(E,\theta) \}^2 q(E,\theta) d\theta}{\int_0^{2\pi}  q(E,\theta) d\theta}.
\end{equation*}
Therefore, it sufficient to show that there exists a positive constant $C$ such that
\begin{equation*}
\int_0^{2\pi} (\tilde{f}(E,\theta)-E_{\nu_{1,E}}[f])^2  q(E,\theta) d\theta \le C \int_0^{2\pi} ( \partial_{\theta} \tilde{f}(E,\theta) )^2  q(E,\theta)^{-1} d\theta 
\end{equation*}
for every $E>0$ and every smooth function $f$. Then, since $d_- \le q(E,\theta) \le d_+$ for all $E>0$ and $\theta$, and
\begin{equation*}
\int_0^{2\pi} (\tilde{f}(E,\theta)-E_{\nu_{1,E}}[f])^2 q(E,\theta) d\theta \le \int_0^{2\pi} \Big(\tilde{f}(E,\theta)- (\int_0^{2\pi} \tilde{f}(E,\theta) d\theta) \Big)^2 q(E,\theta) d\theta
\end{equation*}
holds for every $E>0$ and every smooth function $f$, the desired inequality follows from the Poincar\'e inequality.
\end{proof}
\begin{lemma}\label{lem:3ch3}
There exist positive constants $0 < c \le C <\infty$ such that
  \begin{equation*}
  \label{eq:11}
 c\; E \le  \alpha_i(E)  \le C E
\end{equation*}
for all $E>0$ and for $i=1,2$ where $\alpha_1(E):=E_{\nu_{1,E}}[p_1^2]$ and $\alpha_2(E):=E_{\nu_{1,E}}[V'^2(r_1)]$.
\end{lemma}
\begin{proof}
By the change of variables introduced above, 
\begin{equation*}
E_{\nu_{1,E}}[p_1^2]=\frac{\int_0^{2\pi} 2E \cos \theta^2  q(E,\theta) d\theta}{\int_0^{2\pi} q(E,\theta) d\theta}
\end{equation*}
and it is easy to show that $\frac{d_-}{d_+}E \le E_{\nu_{1,E}}[p_1^2] \le 2E$.
Similarly, 
\begin{equation*}
E_{\nu_{1,E}}[V'^2(r_1)] \le \frac{2}{d_-^2} E_{\nu_{1,E}}[V(r_1)] =\frac{2\int_0^{2\pi} E \sin \theta^2  q(E,\theta) d\theta}{d_-^2 \int_0^{2\pi} q(E,\theta) d\theta} \le \frac{2E}{d_-^2}
\end{equation*}
and
\begin{equation*}
E_{\nu_{1,E}}[V'^2(r_1)] \ge \frac{2}{d_+^2} E_{\nu_{1,E}}[V(r_1)] =\frac{2\int_0^{2\pi} E \sin \theta^2  q(E,\theta) d\theta}{d_+^2 \int_0^{2\pi} q(E,\theta) d\theta}  \ge  \frac{d_-E}{d_+^3}.
\end{equation*}
\end{proof}

The following is the main theorem in this section.
\begin{theorem}
There exists a positive constant $C$ such that
  \begin{equation}
    \label{eq:18}
     \begin{split}
    E_{\nu_{L,E}}[f^2] 
    \le C \sum_{k=1}^{L} E_{\nu_{L,E}}[(X_kf)^2] +
    C L^2 \sum_{k=1}^{L-1} E_{\nu_{L,E}}[(Y_{k,k+1}f)^2]
  \end{split}
  \end{equation}
for every positive integer $L$, every $E>0$, and every smooth function $f$ satisfying $E_{\nu_{L,E}}[f] =0$.
\end{theorem}

\begin{proof}
We start the proof by the usual martingale decomposition. Let $\G_k$ be the $\sigma$-field generated by variables 
$\{ \E_1, \dots, \E_k, p_{k+1}, r_{k+1}, \dots, p_L, r_L\}$. Define $f_k := E_{\nu_{L,E}}[f|\G_k]$ for $k=0,1,\cdots, L$. Note that $f_0 = f$ and $f_L = f_L(\E_1, \dots,\E_L)$. Then, we obtain
\begin{equation*}
  \label{eq:5}
  E_{\nu_{L,E}}[f^2] = \sum_{k=0}^{L-1} E_{\nu_{L,E}}[(f_k - f_{k+1})^2] +
  E_{\nu_{L,E}}[f_L^2].   
\end{equation*}
We analyze each term separately.

By Lemma \ref{lem:2ch3}, for any $k$
\begin{equation*}
  \label{eq:6}
  E_{\nu_{L,E}}[(f_k - f_{k+1})^2|\G_k] \le C  
  E_{\nu_{L,E}}[(X_{k+1}f_{k})^2|\G_k] 
\end{equation*}
and therefore we have
\begin{equation*}
  \label{eq:7}
  \begin{split}
    E_{\nu_{L,E}}[f^2] \le C \sum_{k=1}^{L}E_{\nu_{L,E}}[(X_kf_{k-1})^2] +
    E_{\nu_{L,E}}[f_L^2]\\
    \le C \sum_{k=1}^{L} E_{\nu_{L,E}}[(X_kf)^2] +
    E_{\nu_{L,E}}[f_L^2].
  \end{split}
\end{equation*}
So we are left to estimate $E_{\nu_{L,E}}[f_L^2]$ in terms of the
Dirichlet form \linebreak $\sum_{k=1}^{L-1} E_{\nu_{L,E}} [(Y_{k,k+1}f_L)^2]$.

Observe that $Y_{k,k+1} f_L = p_k V'(r_{k+1}) \left(\partial_{\E_{k}} - \partial_{\E_{k+1}} \right) f_L (\E_1,\dots,\E_L)$.
Since $\nu_{L,E}$ is the conditional probability of the product measure $\nu^{L}_\beta$,
\begin{equation*}
  \label{eq:9}
    E_{\nu_{L,E}}[p_k^2V'(r_{k+1})^2|\G_L] = E_{\nu_{1,\E_k}}[p^2]
    E_{\nu_{1,\E_{k+1}}} [V'^2(r)]
    = \alpha_1(\E_{k})\alpha_2(\E_{k+1}). 
\end{equation*}

By Lemma \ref{lem:3ch3}, the Dirichlet form  $\sum_{k=1}^{L-1}  E_{\nu_{L,E}}[
(Y_{k,k+1}f_L)^2]$, is equivalent to 
\begin{equation*}
  \label{eq:12}
   \sum_{k=1}^{L-1} E_{\nu_{L,E}} \left[ \E_{k}\E_{k+1} \left\{ \left(\partial_{\E_{k}}
    - \partial_{\E_{k+1}} \right) f_L \right\}^2\right]. 
\end{equation*}

Now the problem is reduced to the estimates of the spectral gap for
the generator of an autonomous energy dynamics, depending only on
variables $\E_1, \dots, \E_L$. Since we can write the probability
distribution $\nu_{L,E}(\cdot|\G_L)$ on  
$\{(\E_1, \dots, \E_L): \sum_i \E_i = LE\}$ as the product measure
$\prod_{i=1}^L h_1(x_i) dx_i$ (or $\prod_{i=1}^L h_{\beta}(x_i) dx_i$
for any $\beta$) conditioned on the same surface $\{\sum_{i=1}^L x_i =
EL\}$,   
Theorem \ref{thm:sgenergy} in the next subsection completes the proof.
\end{proof}

\subsection{Spectral gap for the energy dynamics}
\label{sec:spectral-gap-energy}
Consider the product measure $\prod_{i=1}^L h_1(x_i) dx_i$ on $\R_+^L$ and $d\mu_{L,E}$ the 
conditional distribution of it on the surface 
$\Sigma_{L,E}=\{\sum_{i=1}^L x_i = LE\}$.
We have the following expression
$$
d\mu_{L,E} = \prod_{i=1}^L h(x_i) d\lambda_{L,E}(x_1, \dots, x_L)
$$
where $d\lambda_{L,E}$ is the uniform measure on the surface 
$\Sigma_{L,E}$.
\begin{theorem}\label{thm:sgenergy}
There exists a positive constant $C$ such that
\begin{equation*}
  \label{eq:14}
  E_{\mu_{L,E}}[g^2] \le CL^2 \sum_{i=1}^{L-1} E_{\mu_{L,E}}\left[x_i x_{i+1}
  \left(\partial_{x_i} g -\partial_{x_{i+1}}g\right)^2\right]
\end{equation*}
for every positive integer $L$, every $E>0$ and every smooth function $g:\Sigma_{L,E} \to \R$ satisfying $E_{\mu_{L,E}}[g] = 0$.
\end{theorem}

To prove this, we first refer Caputo's result \cite{Ca08}. Recall that $\delta_-/ \delta_+ e^{-x} \le  h_1(x) \le \delta_+/ \delta_- e^{-x}$. Then, we can define a bounded measurable function $b(x)$ satisfying $\int_{0}^{\infty} dx \exp(-x+b(x)) =1$ by $h_1(x)=\exp(-x+b(x))$. The product measure $\prod_{i=1}^L h_1(x_i) dx_i$ is in the class studied in Example 3.1 of \cite{Ca08}. 

Let $E_{i,j}$ and $D_{i,j}$ be operators
defined by $E_{i,j}f=E_{\mu_{L,E}}[f | \F_{i,j}]$ and $D_{i,j}f=E_{i,j}f-f$ where $\F_{i,j}$ is the $\sigma$-algebra generated by variables $\{x_k \}_{k \neq i,j}$. From Theorem 3.1 and the argument of Example 3.1 in \cite{Ca08}, the next lemma follows easily. 
\begin{lemma}[Caputo, \cite{Ca08}]
  \label{Caputo}
If $\frac{4}{9}e^{-16|b|_{\infty}} > \frac{1}{3}$, then there exists a positive constant $C$ such that
  \begin{equation*}
      E_{\mu_{L,E}}[g^2] \le \frac CL \sum_{i, j=1}^L E_{\mu_{L,E}}[(D_{i,j}g)^2]
 \end{equation*}
for every $E>0$, every positive integer $L$ and every smooth function $g:\Sigma_{L,E} \to \R$ satisfying $E_{\mu_{L,E}}[g] = 0$.
\end{lemma}
\begin{remark}
Since $|b|_{\infty} \le \log(\frac{\delta_+}{\delta_-})$, the condition $\delta_-/ \delta_+ > (3/4)^{1/16}$ is a sufficient condition for the assumption of Lemma \ref{Caputo}.
\end{remark}
Next, we show that we can take a telescopic sum.
\begin{lemma}
  \label{telescopic}
There exists a positive constant $C$ such that
  \begin{equation*}
      \frac 1L \sum_{i, j=1}^L E_{\mu_{L,E}}[(D_{i,j}g)^2] 
       \le CL^2 \sum_{i=1}^{L-1} E_{\mu_{L,E}}[(D_{i,i+1}g)^2]
 \end{equation*}
  for every $E>0$, every positive integer $L$ and every smooth function $g:\Sigma_{L,E} \to \R$.
 \end{lemma}

\begin{proof}
First, we rewrite the term $E_{i,j}g$ in an integral form:
\begin{equation*}
E_{i,j}g (x)= \frac{1}{\Xi_{x_i+x_j}} \int^1_0  g(R_{i,j}^{t} x)h((x_i+x_j) t)
h((x_i+x_j)(1-t)) dt
\end{equation*}
where $\Xi_a= \int^1_0  h(a t)
h(a (1-t)) dt$ and $R_{i,j}^{t}x \in \R_+^L$ is a configuration defined by
\begin{equation*}  
 (R_{i,j}^t x )_k = \begin{cases}
	x_k  & \text{if $k \neq i,j$,} \\
	(x_i+x_j)t & \text{if $k=i$,} \\
	(x_i+x_j) (1-t)  & \text{if $k=j$.} 
\end{cases}
\end{equation*}

Then, by Schwarz's inequality we have 
\begin{equation*}
 \begin{split}
(D_{i,j}g (x) ) ^2 &= ( \frac{1}{\Xi_{x_i+x_j}} \int^1_0  \{ g(R_{i,j}^{t} x)-g(x) \} h((x_i+x_j) t) h((x_i+x_j)(1-t)) dt )^2 \\
    & \le \frac{1}{\Xi_{x_i+x_j}} \int^1_0  \{ g(R_{i,j}^{t} x)-g(x) \}^2 h((x_i+x_j) t) h((x_i+x_j)(1-t)) dt. 
\end{split}
\end{equation*}

Now, we introduce operators $\pi^{i,j}$, $\sigma^{i,j}$ and $\tilde{\sigma}^{i,j}: \R_+^L \to \R_+^L$ for $i <j$ as 
\begin{equation*}  
 (\pi^{i,j} x )_k = \begin{cases}
	x_k  & \text{if $k \neq i,j$,} \\
	x_j & \text{if $k=i$,} \\
	x_i  & \text{if $k=j,$} 
\end{cases}
\end{equation*}
$\sigma^{i,j}:= \pi^{j-1,j} \circ \pi^{j-2,j-1} \cdots \circ \pi^{i,i+1}$ and $\tilde{\sigma}^{i,j}:= \pi^{i,i+1} \circ \pi^{i+1,i+2} \cdots \circ \pi^{j-1,j}$.
With these notations, for any $i<j$, we can rewrite the term $g(R_{i,j}^{t} x)-g(x)$ as 
\begin{equation*}
\begin{split}
g(R_{i,j}^{t} x)-g(x) & = \{g(\tilde{\sigma}^{i,j-1}( R_{j-1,j}^{t}( \sigma^{i,j-1}x )))-g(R_{j-1,j}^{t}( \sigma^{i,j-1}x ))\} \\
 &+ \{g(R_{j-1,j}^{t}( \sigma^{i,j-1}x ))-g( \sigma^{i,j-1}x )\} + \{g( \sigma^{i,j-1}x )-g(x)\}.
\end{split}
\end{equation*}
Therefore, we can bound the term $E_{\mu_{L,E}} [ (D_{i,j}g(x))^2 ]$ from above by
\begin{equation}
\label{schwarz}
\begin{split}
3 E_{\mu_{L,E}} [ \frac{1}{\Xi_{x_i+x_j}} & \int^1_0 \{g(\tilde{\sigma}^{i,j-1}( R_{j-1,j}^{t}( \sigma^{i,j-1}x )))-g(R_{j-1,j}^{t}( \sigma^{i,j-1}x ))\}^2 \\
     &  h((x_i+x_j) t) h((x_i+x_j)(1-t)) dt ] \\
    +  3 E_{\mu_{L,E}} [ \frac{1}{\Xi_{x_i+x_j}} &\int^1_0  \{g(R_{j-1,j}^{t}( \sigma^{i,j-1}x ))  -g( \sigma^{i,j-1}x )\}^2 \\
    & h((x_i+x_j) t) h((x_i+x_j)(1-t)) dt]\\
    + 3 E_{\mu_{L,E}} [ \frac{1}{\Xi_{x_i+x_j}} & \int^1_0   \{g(\sigma^{i,j-1} x )   -g( x )\}^2 h((x_i+x_j) t) h((x_i+x_j)(1-t)) dt].
\end{split}
\end{equation}
We estimate three terms separately.
The last term of equation (\ref{schwarz}) is equal to
\begin{equation*}
3 E_{\mu_{L,E}} [ \{g(\sigma^{i,j-1}x )-g( x )\}^2 ]
\end{equation*}
and therefore bounded from above by
\begin{equation*}
3L \sum_{k=i}^{j-2}E_{\mu_{L,E}} [ \{g(\pi^{k,k+1}x )-g( x )\}^2 ].
\end{equation*}
By simple computations, we obtain that 
\begin{equation*}
\begin{split}
& E_{\mu_{L,E}} [ \{g(\pi^{k,k+1}x )-g( x )\}^2 ] \\
&=  E_{\mu_{L,E}} [ \{g(\pi^{k,k+1}x )-(E_{k.k+1}g)(\pi^{k,k+1}x)+(E_{k.k+1}g)(x)- g( x )\}^2 ] \\
& \le 2 E_{\mu_{L,E}} [ \{g(\pi^{k,k+1}x )-(E_{k.k+1}g)(\pi^{k,k+1}x) \}^2] + 2E_{\mu_{L,E}}  [\{(E_{k.k+1}g)(x)- g( x )\}^2 ]  \\
& = 4E_{\mu_{L,E}}  [(D_{k.k+1}g)^2 ].
\end{split}
\end{equation*}
By the change of variable with $y= \sigma^{i,j-1}x$, we can rewrite the second term of equation (\ref{schwarz}) as
\begin{equation*}
\begin{split}
 3 E_{\mu_{L,E}} & [ \frac{1}{\Xi_{y_{j-1}+y_j}} \int^1_0   \{g(R_{j-1,j}^{t}y)-g( y)\}^2 h((y_{j-1}+y_j) t) h((y_{j-1}+y_j)(1-t)) dt] \\
 & = 3 E_{\mu_{L,E}}[E_{j,j+1}(g^2)-2gE_{j.j+1}g+g^2] \\
 &= 6E_{\mu_{L,E}}[g^2-(E_{j.j+1}g)^2]= 6 E_{\mu_{L,E}} [(D_{j,j+1}g)^2].
 \end{split}
\end{equation*}
Similarly, the first term of equation (\ref{schwarz}) is rewritten as 
\begin{equation*}
\begin{split}
3 E_{\mu_{L,E}} & [ \frac{1}{\Xi_{y_{j-1}+y_j}}  \int^1_0  \{g(\tilde{\sigma}^{i,j-1}( R_{j-1,j}^{t}y))-g(R_{j-1,j}^{t}y )\}^2  \\ 
 & h((y_{j-1}+y_j) t) h((y_{j-1}+y_j)(1-t)) dt ] \\
= 3 E_{\mu_{L,E}} & [ E_{j,j+1}( \{g \circ \tilde{\sigma}^{i,j-1}-g \}^2 )  ] = 3 E_{\mu_{L,E}} [ \{g \circ \tilde{\sigma}^{i,j-1}-g \}^2  ] .
 \end{split}
\end{equation*}
In the same way as the first term of (\ref{schwarz}),
it is bounded from above by \linebreak $12L \sum_{k=i}^{j-2} E_{\mu_{L,E}} [ (D_{k,k+1}g)^2  ] $.
Therefore, we complete the proof.
\end{proof}

\begin{lemma}\label{lem:2state}
There exists a constant $C$ such that
  \begin{equation} \label{eq:twosg}
      E_{\mu_{2,E}} [(D_{1,2}g)^2]  \le C E_{\mu_{2,E}} [ x_1 x_2 (\partial_{x_1}g-\partial_{x_2}g)^2 ].
 \end{equation}
 for every $E>0$ and every smooth function $g:\Sigma_{2,E} \to \R$. 
 \end{lemma}

\begin{proof}
Since the both sides of (\ref{eq:twosg}) do not change if we replace $g$ with $g+a$ for any constant $a$, it is sufficient to show that the inequality holds for every smooth function $g:\Sigma_{2,E} \to \R$ satisfying $E_{\lambda_{2,E}} [ g ]=0$. In particular, since $E_{\mu_{2,E}} [(D_{1,2}g)^2] \le E_{\mu_{2,E}} [g^2]$, it is sufficient to show that
\begin{equation*} 
      E_{\mu_{2,E}} [g^2]  \le C E_{\mu_{2,E}} [ x_1 x_2 (\partial_{x_1}g-\partial_{x_2}g)^2 ].
 \end{equation*}
Note that for any positive function $f:\Sigma_{2,E} \to \R_+$ and for
any $E>0$,  
\begin{equation*}
(\frac{\delta_-} {\delta_+}) ^4 E_{\mu_{2,E}} [ f ] \le E_{\lambda_{2,E}} [ f ]  \le ( \frac{\delta_+} {\delta_-})^4 E_{\mu_{2,E}} [ f ].
 \end{equation*}
In fact, 
\begin{equation*}
 E_{\mu_{2,E}} [ f ]= \frac{1}{\Xi_E} \int^1_0  f (tE, (1-t)E )h(tE) h((1-t)E) dt
 \end{equation*}
where $\Xi_E=\int^1_0 h(tE) h((1-t)E) dt$. Then, since $\frac{\delta_-} {\delta_+} \le h(x) \le \frac{\delta_+} {\delta_-}$, the above estimate holds.
Now, all we have to show is that, there exists a constant $C$ such that  
 \begin{equation}\label{eq:2pointsg}
      E_{\lambda_{2,E}} [g^2]  \le C E_{\lambda_{2,E}} [ x_1 x_2
      (\partial_{x_1}g-\partial_{x_2}g)^2 ] 
 \end{equation}
 for every $E>0$ and every smooth function $g:\Sigma_{2,E} \to \R$ satisfying $E_{\lambda_{2,E}} [ g ]=0$. 
 By the definition of $\lambda_{2,E}$, the inequality (\ref{eq:2pointsg}) is rewritten as
  \begin{equation*}
        \int_0^E g(t)^2 dt \le  C \int_0^E t(E-t) g'(t)^2 dt
 \end{equation*}
 and by a suitable change of variable, the problem is reduced to the case with $E=1$.
Applying Schwarz inequality and changing the order of integration repeatedly, we have
   \begin{align*}
        \int_0^1 & g(t)^2 dt =  \int^1_0 \int_0^t \{ \int^t_s g'(r) dr \}^2 ds dt   \le  \int^1_0 \int_0^t  (t-s)  \int^t_s g'(r)^2 dr  ds dt \\
        & = \int^1_0 \int_0^t  (tr-\frac{r^2}{2})   g'(r)^2  dr dt = \frac{1}{2} \int^1_0   g'(r)^2  r(1-r) dr.
  \end{align*} 
 
 \end{proof} 

\begin{lemma}
There exists a positive constant $C$ such that
\begin{equation*}
  \begin{split}
     E_{\mu_{L,E}}[(D_{i,i+1}g)^2] \le C E_{\mu_{L,E}}[ x_ix_{i+1} (\partial_{x_i}g-\partial_{x_{i+1}}g)^2]
    \end{split}
 \end{equation*}
 for every positive integer $L$, $i =1, \dots, L-1$, and every smooth function $g:\Sigma_{L,E} \to \R$.
 \end{lemma}

\begin{proof}
By Lemma \ref{lem:2state}, 
\begin{equation*}
E_{\mu_{L,E}}[(D_{i,i+1}g)^2 | \F_{i,i+1}] \le C E_{\mu_{L,E}}[ x_ix_{i+1} (\partial_{x_i}g-\partial_{x_{i+1}}g)^2 | \F_{i,i+1}]
\end{equation*}
holds. Then, by taking the expectation, we complete the proof.
\end{proof}

\section{Lie Algebra}
\label{sec:lie-algebra}

We prove here that $\text{Lie}\{X_i, Y_{i,i+1}, i=1, \dots, N\}$ generates
the all tangent space of $\Sigma_{N,E}$. We have used this property in
section \ref{sec:closedformch2}, for the characterization of the
finite dimensional closed forms.

We have
\begin{equation*}
  [X_{j+1}, Y_{j,j+1}] = V''(r_{j+1}) Z_{j,j+1}
\end{equation*}
where
\begin{equation*}
  Z_{i,j} = p_i \partial_{p_j} - p_j \partial_{p_i} 
\end{equation*}
Since $[ Z_{j,j+1}, Z_{j+1,j+2}] = Z_{j,j+2}$ and $V''(r) > \delta
>0$,
we have that $Z_{i,j} \in \text{Lie}\{X_i, Y_{i,i+1}, i=1, \dots, N\}$ for
any $i$ and $j$.

On the other hand $[Z_{j,i}, X_j] = Y_{i,j}$, and we have enough
vector fields to generate the all tangent space.

\begin{remark}
By the above argument, it is obvious that $\text{Lie}\{ \{ X_i, i=1, \dots, N \} \{ Y_{i,i+1}, i=1, \dots, N-1 \} \}$ also generates the all tangent space of $\Sigma_{N,E}$.
\end{remark}

\end{document}